\documentclass[twoside, leqno, twocolumn]{article}

\usepackage[letterpaper]{geometry}
\usepackage{ltexpprt}
\usepackage{hyperref}

\newif\ifcomments
\commentstrue

\usepackage{pgfplots}
\usepackage{pgfplotstable}
\pgfplotsset{compat=1.8}
\usepackage{subcaption}
\usepackage{xspace} 
\usepackage[subtle, mathspacing=normal]{savetrees}

\usepackage{xcolor}
\usepackage{amsmath}
\usepackage{amsfonts}
\usepackage{amssymb}
\usepackage{algorithm}
\usepackage{algorithmicx}
\usepackage{float}
\usepackage{siunitx}
\usepackage{graphicx}
\usepackage{thmtools}
\usepackage{thm-restate}

\usepackage[normalem]{ulem} 

\usepackage{thm-restate}
\usepackage{blkarray}
\usepackage{balance}  
\usepackage{booktabs} 

\usepackage{marginnote} 
\usepackage{url}
\usepackage{enumitem}
\graphicspath{{./fig/}}
\usepackage{mathtools}

\usepackage[noend]{algpseudocode}
\usepackage{comment}
\usepackage{color}
\usepackage{cleveref}
\usepackage{tikz}
\usepackage{todonotes} 
\usetikzlibrary{fadings}

\usepackage{dirtytalk}
\usepackage{array}

\usepackage[para,online,flushleft]{threeparttable}

\newcommand{\punt}[1]{}

\newtheorem{problem}{Problem}
\newtheorem{claim}{Claim}

\makeatletter
\def\@copyrightspace{\relax}
\makeatother

\newcommand{\defn}[1]       {{\textit{\textbf{\boldmath #1}}}\xspace}

\DeclareMathOperator\polylog{polylog}

\newcommand{\stream}{S}

\newcommand{\graph}{\mathcal{G}}
\newcommand{\nodes}{\mathcal{V}}
\newcommand{\edges}{\mathcal{E}}
\newcommand{\nodesize}{V}
\newcommand{\edgesize}{E}
\newcommand{\graphstream}{S}
\newcommand{\streamlength}{N}
\newcommand{\sketch}{\mathcal{S}}

\newcommand{\cachesize}{C}
\newcommand{\cachelinesize}{\mathcal{L}}

\newcommand{\streamelement}{s}
\newcommand{\indexsubset}{b}

\newcommand{\veclength}{n}

\newcommand{\charvec}{f}

\newcommand{\sketchsize}{\phi}

\newcommand{\Boruvka}{Bor\r{u}vka\xspace}

\newcommand{\sysname}{\textsc{Landscape}\xspace}
\newcommand{\graphzep}{\textsc{GraphZeppelin}\xspace}

\newcommand{\sketchname}{\textsc{CameoSketch}\xspace}
\newcommand{\sketchnames}{\textsc{CameoSketches}\xspace}
\newcommand{\cubesketch}{\textsc{CubeSketch}\xspace}

\newcommand{\treename}{pipeline hypertree\xspace}
\newcommand{\Treename}{Pipeline Hypertree\xspace}

\newcommand{\dsuname}{\textsc{GreedyCC}\xspace}

\newcommand{\ie}{i.e.,\xspace}

\newcommand{\boundaryidx}{\beta}
\newcommand{\buck}{b}

\newcommand{\etal}{\text{et al}.\xspace}

\date{}

\newcommand{\namedcomment}[3]{{\sf \color{#2} #1: #3}}

\newcommand{\mab}[1]{\namedcomment{mab}{red}{#1}}
\newcommand{\mfc}[1]{\namedcomment{mfc}{purple}{#1}}
\newcommand{\david}[1]{\namedcomment{david}{red}{#1}}
\newcommand{\ahmed}[1]{\namedcomment{ahmed}{blue}{#1}}
\newcommand{\victor}[1]{\namedcomment{victor}{green}{#1}}
\newcommand{\evan}[1]{\namedcomment{evan}{orange}{#1}}

\newcommand{\gil}[1]{\namedcomment{gil}{green}{#1}}

\renewcommand{\mab}[1]{\todo[size=\tiny,color=green!40]{MAB: #1}}
\renewcommand{\mfc}[1]{\todo[size=\tiny,color=green!40]{MFC: #1}}
\renewcommand{\david}[1]{\todo[size=\tiny]{David: #1}}
\renewcommand{\ahmed}[1]{\todo[size=\tiny,color=yellow]{Ahmed: #1}}
\renewcommand{\victor}[1]{\todo[size=\tiny,color=yellow]{Victor: #1}}
\renewcommand{\evan}[1]{\todo[size=\tiny,color=red!40]{Evan: #1}}

\renewcommand{\gil}[1]{\todo[size=\tiny,color=orange!40]{Gil: #1}}

\newcommand{\fixme}[1]{\todo[size=\tiny]{#1}}

\newcommand{\inline}[1]{\todo[inline,color=yellow,size=\tiny]{#1}}

\iffalse
\else

\renewcommand{\mab}[1]{}
\renewcommand{\mfc}[1]{}
\renewcommand{\david}[1]{}
\renewcommand{\ahmed}[1]{}
\renewcommand{\victor}[1]{}
\renewcommand{\evan}[1]{}
\renewcommand{\gil}[1]{}
\renewcommand{\fixme}[1]{}
\renewcommand{\inline}[1]{}

\fi 

\newcommand\hcancel[2][black]{\setbox0=\hbox{$#2$}%
\rlap{\raisebox{.45\ht0}{\textcolor{#1}{\rule{\wd0}{1pt}}}}#2}

\renewcommand{\epsilon}{\varepsilon}

\renewcommand{\eqref}[1]          {Eq.~\ref{eq:#1}}

\newcolumntype{P}[1]{>{\centering\arraybackslash}p{#1}}

\definecolor{bg}{rgb}{0.95,0.95,0.95}

\newcommand{\cameoupdatethm}{\sketchname is an $\ell_0$-sampler that, for vector $x \in \mathbb{Z}_2^n$, uses $O(\log^2(n)\log(1/\delta))$ space, has worst-case update time $O(\log(1/\delta))$, and succeeds with probability $1-\delta$.}

\newcommand{\restatecameoupdatethm}{\vspace{\baselineskip} Recall \textsc{Theorem~\ref{thm:cameo_update}} \emph{\cameoupdatethm}}

\newcommand{\cameospacethm}{Using $3$-wise independent hash functions, \sketchname requires $8\log_{3}(1/\delta) (\log n + 5)$ bytes of space to return a nonzero element of a length $n < 2^{64}$ input vector w/p at least $1-\delta$.}

\newcommand{\restatecameospacethm}{\vspace{\baselineskip} Recall \textsc{Theorem~\ref{thm:cameo_cols}} \emph{\cameospacethm}}

\renewcommand{\paragraph}[1]{\smallskip\noindent{\bf \boldmath #1}}

\date{}

\begin{document}
\sloppy
\newcommand\relatedversion{}
\renewcommand\relatedversion{} 


\title{\Large Exploring the \sysname of Distributed Graph Sketching\relatedversion}



\author{David Tench \thanks{Lawrence Berkeley National Laboratory}
\and Evan T. West \thanks{Stony Brook University}
\and Kenny Zhang \thanks{Massachusetts Institute of Technology}
\and Michael A. Bender \footnotemark[2]
\and Daniel DeLayo \footnotemark[2]
\and Mart\'\i{}n Farach-Colton \thanks{New York University}
\and Gilvir Gill \footnotemark[2]
\and Tyler Seip \thanks{MongoDB}
\and Victor Zhang \thanks{Meta Platforms}
}


\maketitle

\fancyfoot[R]{\scriptsize{Copyright \textcopyright\ 2025 by SIAM\\
Unauthorized reproduction of this article is prohibited}}

\begin{abstract}

Recent work has initiated the study of dense graph processing using graph sketching methods, which drastically reduce space costs by lossily compressing information about the input graph. In this paper, we explore the strange and surprising performance landscape of sketching algorithms. We highlight both their surprising advantages for processing dense graphs that were previously prohibitively expensive to study, as well as the current limitations of the technique. Most notably, we show how sketching can avoid bottlenecks that limit conventional graph processing methods. 

Single-machine streaming graph processing systems are typically bottlenecked by CPU performance, and distributed graph processing systems are typically bottlenecked by network latency. We present \sysname, a distributed graph-stream processing system that uses linear sketching to distribute the CPU work of computing graph properties to distributed workers with no need for worker-to-worker communication. As a result, it overcomes the CPU and network bottlenecks that limit other systems. In fact, for the connected components problem, \sysname achieves a stream ingestion rate one-fourth that of maximum sustained RAM bandwidth, and is four times faster than random access RAM bandwidth. Additionally, we prove that for any sequence of graph updates and queries \sysname consumes at most a constant factor more network bandwidth than is required to receive the input stream. We show that this system can ingest up to 332 million stream updates per second on a graph with $2^{17}$ vertices.  We show that it scales well with more distributed compute power: given a cluster of 40 distributed worker machines, it can ingest updates 35 times as fast as with 1 distributed worker machine. Graph sketching algorithms tend to incur high computational costs when answering queries; to address this \sysname uses heuristics to reduce its query latency by up to four orders of magnitude over the prior state of the art.
\end{abstract}



The full version of the paper can be accessed at \url{https://arxiv.org/abs/2410.07518}

Our code and experiments can be found at \url{https://github.com/GraphStreamingProject/Landscape} and \url{https://doi.org/10.5281/zenodo.13845156}

\section{Introduction}

Computing connected components is a fundamental graph-processing task
with uses throughout computer science and engineering. It has applications in relational databases~\cite{sigmodconn1},
scientific computing~\cite{scientific_computing1,scientific_computing2}, pattern recognition~\cite{pattern1,pattern2}, graph partitioning~\cite{partition1,partition2}, random walks~\cite{randomwalks}, social network community detection~\cite{lee2014social}, graph compression~\cite{compression1,compression2}, medical imaging~\cite{tumor}, flow simulation~\cite{bioinformatics1}, genomics~\cite{metagenomics1,metagenomics2}, identifying protein families~\cite{protein1,bioinformatics2}, microbiology~\cite{cell}, and object recognition~\cite{object}.  Strictly harder problems such as edge/vertex connectivity, shortest paths, and $k$-cores often use it as a subroutine. Connected Components is also used as a heuristic for clustering problems~\cite{clustering1,clustering2,clustering3,clustering4,clustering5,clustering6}, pathfinding algorithms (such as Djikstra and $A^*$), and some minimum spanning tree algorithms. A survey by Sahu \etal~\cite{ubiquitous} of database applications of graph algorithms reports that, for both practitioners and academic researchers, connected components was the most frequently performed computation from a list of 13 fundamental graph problems that includes shortest paths, triangle counting, and minimum spanning trees.

Computing the minimum cut of a graph (or equivalently its edge connectivity) is a closely related problem to connected components. It has applications in clusterings on similarity graphs~\cite{similarityclustering1,similarityclustering2}, community detection~\cite{communitydetection}, graph drawing~\cite{kant1993algorithms}, network reliability~\cite{ramanathan1987counting,karger1995randomized}, and VLSI design~\cite{krishnamurthy1984improved}.

The task of computing connected components or minimum cut becomes more difficult when graphs are \defn{dynamic}, meaning the edge set changes over time subject to a stream of edge insertions and deletions, and this task becomes harder still when the graphs are very large. Applications using dynamic graphs include identifying objects from a video feed rather than a static image~\cite{movingobject} or tracking communities in social networks that change as users add or delete friends~\cite{dynamic_social,dynamic_web}. Applications on large graphs include metagenome assembly tasks that may include gene databases with hundreds of millions of entries with complex relations~\cite{metagenomics1}, and large-scale clustering (a common machine learning challenge~\cite{clustering3}).  And of course graphs can be both large and dynamic.  Indeed, Sahu \etal's~\cite{ubiquitous} database applications survey reports that a majority of industry respondents work with large graphs ($> 1$ million vertices or $> 1$ billion edges) and a majority work with dynamic graphs.

\begin{figure}
    \centering
    \includegraphics[width=.5\textwidth]{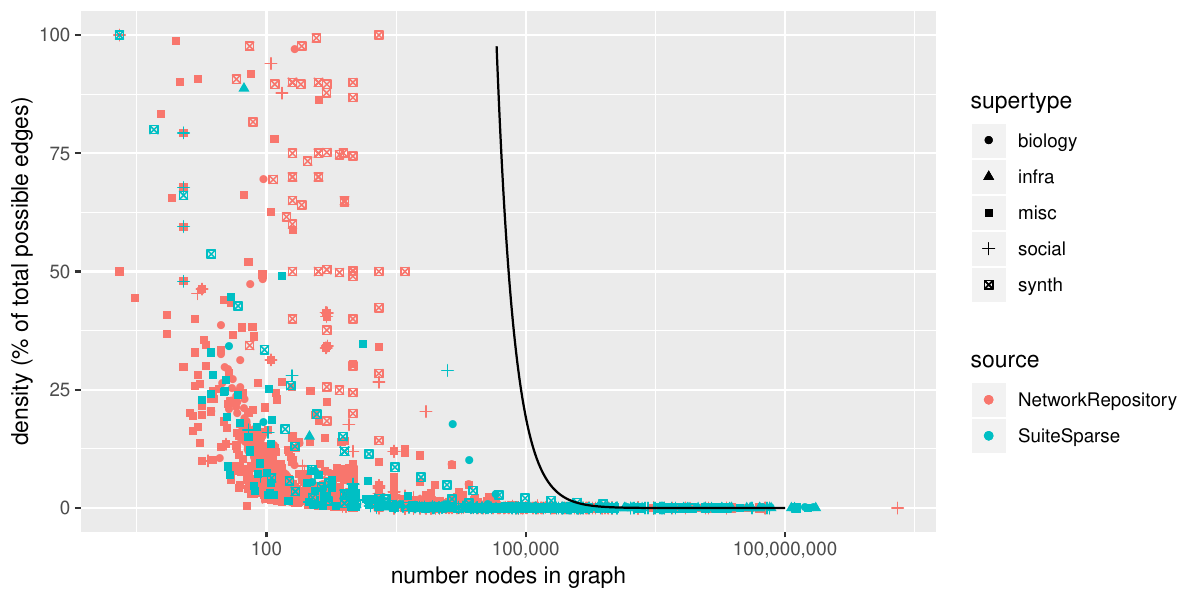}
    \caption{Graphs studied in academic works exhibit a selection effect. Any point to the left of the dark line indicates a dataset which can be represented as an adjacency list in 16GB of RAM.}
    \label{fig:dense}
\end{figure}

\paragraph{Dense-graph processing.} 
The task of computing connected components is especially
difficult for dynamic \defn{dense graphs}. The conventional wisdom is that massive graphs are always sparse, meaning that they have few edges per vertex.
Tench \etal ~\cite{graphzeppelin} contend that instead large, dense graphs do not appear in academic publications due to a selection effect: since we lack the tools to process these graphs, they are not studied. We expand on their survey of graph datasets to further support this claim. Figure~\ref{fig:dense} plots all graph datasets from the NetworkRepository~\cite{netrepo} and SuiteSparse~\cite{suitesparse} collections. Note that nearly all the graphs can be stored as an adjacency lists using less than 16GB, and this pattern holds across repositories and across different types of graph dataset - including a variety of biological data, social networks, and infrastructure(e.g., road and computer networks). See Appendix~\ref{app:dense} for an expanded analysis which shows this selection effect even more strongly.

Despite this effect, there is evidence of dense graphs emerging in practical applications. Tench \etal note that ``Facebook works with graphs with 40 million
nodes and 360 billion edges. These graphs are processed at great
cost on large high-performance clusters, and are consequently not
released for general study''~\cite{graphzeppelin}. As another example, bipartite projection methods~\cite{neal2021comparing}, commonly used in social sciences and bioinformatics, naturally generate large, dense graphs. Current techniques require storing these graphs in RAM, limiting the size of datasets analyzed this way --- an explicit example of the selection effect~\cite{olson2013navigating}. The only way to conclusively determine whether there exist more applications that would benefit from dense-graph processing systems is to build such systems and see what applications emerge.


Tench \etal~\cite{graphzeppelin}  demonstrate that \defn{linear
sketching} techniques~\cite{mcgregor2014graph,Ahn2012,AhnGM12b,GuhaMT15}
can be used to process large, dense, and dynamic graphs. Linear sketching saves the most space when graphs are dense---this is because
the size of a connectivity sketch of a $\nodesize$-vertex graph is $O(\nodesize \polylog\nodesize)$ and therefore
is independent of the number of edges. This algorithm is representative of a large number of graph-sketching algorithms~\cite{mcgregor2014graph,Ahn2012,AhnGM12b,KapralovLMMS13,GuhaMT15,assadi2023tight} in the database theory literature that all share the same general structure.

However, the price for the small space of the sketches is high CPU cost: processing each update requires $O(\log^2\nodesize)$ work, and the constants hidden by the asymptotic notation are large. Concretely, for a million-vertex graph, processing each edge \defn{update} (an insertion or deletion) requires evaluating roughly 500 hash functions (in addition to other costs). Tench \etal's implementation, \graphzep, introduces some techniques for mitigating this high computational cost but ultimately their implementation is bottlenecked by CPU.
\paragraph{This paper.} We design and implement
distributed sketching algorithms for connected components and k-connectivity (or bounded $k$ mincut) on dynamic graph streams. Using our implementation, we explore the surprising performance landscape of graph-sketching algorithms. Thus, we call this graph streaming system \sysname. 

The peaks and valleys of this performance landscape will seem unfamiliar to graph-processing practitioners. The highest-order observation is that these techniques work well when graphs are dense and work poorly when they're sparse. This is in contrast to most techniques which work better when graphs are sparse.  We also show these sketching algorithms can avoid computational bottlenecks that are unavoidable using conventional graph-processing techniques, but these algorithms also struggle in some cases where traditional algorithms excel. 

We summarize these findings in Claim~\ref{claim:all} in Section~\ref{sec:results}, but for context we first briefly survey the bottlenecks in graph-stream processing and how they affect existing graph-processing systems (which are designed for sparse graphs).

\smallskip


\subsection{Three Bottlenecks in Graph-Stream Processing}
Graph-stream processing systems can be bottlenecked on any of three resources: \defn{space}, \defn{CPU}, and \defn{communication}. Even a single-worker system requires some network communication; at minimum it must receive the input stream from a network link.

Some systems, such as Aspen~\cite{aspen} and Terrace~\cite{terrace}, are bottlenecked on \defn{space}. These systems maintain lossless representations of the input graph in RAM on a single machine, and such lossless representations can be large.  These systems are optimized for processing sparse graphs where this memory burden is less onerous but struggle when processing large, dense graphs.

Other systems, such as \graphzep~\cite{graphzeppelin}, are bottlenecked on \defn{computation} (as we explain earlier).  \graphzep maintains a lossily-compressed graph representation, reducing storage requirements (and improving performance on dense graphs) at the cost of higher CPU load for processing updates and  answering
queries.

Distributed graph-stream processing systems are bottlenecked by network \defn{communication}. 
For dynamic systems~\cite{kickstarter}, the input stream is received at a central node called the \defn{main node} and information about the graph is distributed across the cluster to \defn{worker nodes}. For non-dynamic systems the data is distributed among worker nodes before computation begins. Since the graph data is spread across many worker nodes, and graphs often have poor data locality, most computation on graphs require worker nodes to send lots of information to each other~\cite{kineograph,graphtau,kickstarter,iyer2015celliq,iyer2021tegra,distinger}.  This is an example of a general challenge in distributed database systems.  These systems scale by spreading data among the aggregated memory of nodes in a cluster at the cost of high inter-node communication, which can increase query and transaction latencies~\cite{scalestore,dragojevic2014farm,Dragojevic2015NoCD}.

\begin{table}[t]
    \centering
     \caption{Summary of ingestion bottlenecks. Existing graph systems are bottlenecked by space, CPU, or network communication costs when processing dense graph streams. In contrast, in this paper we present a system that overcomes these three bottlenecks.}
    \begin{threeparttable}
    
    \begin{tabular}{|P{24mm}|| P{18mm}|P{9mm}|P{14mm}|}
    \hline
         & \textbf{space} & \textbf{\small CPU} & \textbf{\small Network} \\
         \hline 
      \small single-machine lossless~\cite{aspen, terrace} & $\star$\newline $\Theta(\nodesize^2)$ & $\star$ $\star$ & N/A  \\
      \hline 
      \small single-machine sketching~\cite{graphzeppelin} & {$\star$ $\star$ $\star$} \newline $\Theta(\nodesize \log^3(\nodesize))$ & $\star$ & N/A  \\
      \hline
      \small distributed 
      lossless 
\cite{kineograph,graphtau,kickstarter,iyer2015celliq,iyer2021tegra,distinger} & $\star$ $\star$  \newline $\Theta(\nodesize^2)$ & $\star$ $\star$  & $\star$ \\
      \hline
      \small distributed sketching 
      (this paper) & {$\star$ $\star$ $\star$}   \newline $\Theta(\nodesize \log^3(\nodesize))$  & $\star$ $\star$  & $\star$ $\star$ $\star$  \\
      \hline
    \end{tabular}
    \begin{tablenotes}
    \item[$\star$] = bottleneck; \item[$\star$ $\star$ ] = good; \item[ $\star$ $\star$ $\star$] = optimal
    \end{tablenotes}
    \end{threeparttable}
    \label{tab:report_card}
\end{table}

Table~\ref{tab:report_card} summarizes the bottlenecks for each of these approaches. 


\paragraph{Objective standard for update performance: RAM bandwidth.} 
One issue with evaluating the performance of a distributed system to compute fully-dynamic (i.e., edges can be inserted or deleted) graph connectivity is that there are no other open-source distributed systems that solve this precise problem (see Appendix~\ref{sec:related}). However, for stream ingestion rate we can compare against objective upper bound.

The update performance of any graph-streaming system is limited by the \defn{data acquisition cost}, that is, the cost for the main node to simply read the entire input stream. Even if we ignore the cost to update the connectivity information in the graph data structures, we still need to read the input. Thus, an objective standard describing the ideal performance is simply the RAM bandwidth. 
In fact, there are two notions of RAM bandwidth: \defn{random-access RAM bandwidth} (the speed at which we can write words to random locations) and 
\defn{sequential-access RAM bandwidth} (the speed at which we can write words to sequential locations). 

Since graphs have notoriously poor data locality, an update rate close to random-access RAM bandwidth is a natural goal for a graph stream-processing system.
Sequential-access RAM bandwidth is truly a bound on the best possible update rate because any stream-processing system must write the input data into memory, that is, it must pay the data acquisition cost.

\section{Results}
\label{sec:results}
In this paper, we build the \sysname graph-processing system, which is optimized for dense, dynamic graphs. We use this system to establish how graph sketching can overcome classical graph processing bottlenecks.

\subsection{The \sysname Graph-Processing System}

We build \sysname, a linear-sketch-based distributed graph-stream processing system that computes connected components and k-connectivity of dynamic graphs. 

\sysname keeps its sketch (a lossily compressed graph representation) on the main node and distributes the CPU work of processing updates. 

Because these linear sketches are small (size $\Theta(\nodesize\log^3\nodesize)$), they fit on a single node, even when the graph is dense.
The work to maintain these sketches can be chunked off into large batches that can be computed independently by worker nodes. 
As a result, \sysname avoids the CPU bottleneck because it has lots of worker nodes to help, and \sysname avoids the communication bottleneck because the communication cost is amortized away by the CPU cost to process a batch. 

In fact, we prove theoretically that for any sequence of graph updates and queries 
\sysname's total communication cost is only a small, user-configurable constant (by default, four) factor larger than the cost for the main node to receive the input stream. 
Additionally, we introduce a new sketching algorithm, \sketchname, which requires only $O(\log\nodesize)$ distributed work per update (compared to $O(\log^2\nodesize)$ for the prior state of the art), which allows the algorithm to scale more rapidly with limited cluster resources.

\paragraph{Performance.}
\sysname achieves the following:
\begin{itemize}[noitemsep,topsep=0pt,leftmargin=*]
    \item \sysname is able to process graph streams only $4.5\times$ slower than the multi-threaded RAM sequential-write bandwidth, the objective upper bound on insertion performance for any system that receives the input stream at the main node. This is more than four times faster than random access RAM bandwidth.
    \item We show that \sysname can ingest up to 332 million stream updates per second on a graph with $2^{17}$ vertices.
    \item We show that it scales well with more distributed compute power: given a cluster of 40 distributed worker machines, it can ingest updates 35 times as fast as with 1 distributed worker machine.
    \item We experimentally verify that \sysname uses at most $4\times$ the network bandwidth required to read the input stream.
    \item \sysname's  \dsuname query heuristic reuses partial information from prior query results, achieving up to a four orders-of-magnitude reduction in query latency.
\end{itemize}

\paragraph{Outperforming lossless representations on dense graphs.}
To put this performance in context, consider a simpler task: maintaining an adjacency matrix of the graph defined by the input stream.\footnote{Note that while adjacency matrices may be compact, they have high query latency. Even disregarding this limitation, we see that adjacency matrices are outperformed by sketching on dense graphs. }
If the graph is dense and edges are random, an adjacency matrix is essentially the space-optimal lossless graph representation. 
We ignore the cost of answering queries, which an adjacency matrix does not efficiently support.

\sysname's graph-sketch representation is smaller than this adjacency matrix even when the input graph has only 310,000 vertices.
Even more interestingly, \sysname's update throughput is also faster than the update throughput of the adjacency-matrix representation---which is just a single bit flip per edge.  
We emphasize: one of the most dramatic advantages of distributed graph sketching is that updates are faster than adjacency-matrix updates even when the entire adjacency matrix fits in RAM. 

\sysname's updates are fast not because they are small (you cannot beat a single bit flip), but because the \sketchnames have good data locality---and the edge updates result in primarily sequential accesses to RAM on the main node, rather than random access.
Said differently, \sysname processes more edge updates per second than it is possible to flip bits in random locations in RAM.

\subsection{Circumventing the Classical Bottlenecks for Graph-Stream Processing}

The performance implications of sketching for dense graph processing are encapsulated in the following claim. Throughout this paper, as we present theoretical and experimental results, we will refer to the element of the claim they support. 

\newpage
\begin{claim}[\textbf{Dense graph processing}] 
\label{claim:all}
~
\evan{Moved to next column}
\begin{enumerate}[leftmargin=*]
\item \label{itm:space} \textbf{Space consumption.} Graph sketches for dynamic, massive,  dense graphs can be maintained so that they use less space than traditional graph-storage methods.
E.g., \sysname is asymptotically space-optimal and its new \sketchname algorithm uses 29\% of the space of the prior state-of-the-art~\cite{graphzeppelin}. 

\item \label{itm:cpu} \textbf{CPU cost.} Graph sketches have high CPU cost to update, but this cost can be distributed away. 
E.g., \sysname's \sketchname algorithm reduces the asymptotic work per update from $O(\log^2 \nodesize)$ (the prior SOTA) to $O(\log\nodesize)$. We show that this yields a $7\times$ increase to update throughput in experiments. 

\item \label{itm:comm} \textbf{Communication costs.} The above distribution of CPU work can be done with nearly optimal communication: specifically, the total communication cost is only a small constant number of times larger than the data acquisition cost. 
We prove this theoretically and validate it experimentally.

\item \label{itm:ram}\textbf{Stream-ingestion can be blindingly fast--nearly the universal speed limit.} 
%
%
There is a universal speed limit for stream ingestion, which is simply the cost to write the stream sequentially into RAM. A graph-sketch based system for connectivity and k-connectivity can match this bound within a remarkably small constant factor. \sysname ingests graph data at a rate that is within a factor $4$ of sequential RAM bandwidth. 
\end{enumerate}
\end{claim}
Table~\ref{tab:report_card} summarizes the strengths of the sketching approach: \sysname ingestion is not bottlenecked on space, CPU, or communication.

The following observations are nearly immediate and unsurprising, and we include them for completeness. 

\textbf{Queries.} Similar to the trade-off between reading and writing in (nongraph) databases~\cite{ONeilChGa96,BenderFaJa15,SearsCaBr08,BenderFaFi07}, ingesting a graph stream and processing it into a sketch is faster than querying the sketch. \sysname answers each query in single digit seconds, even when there are 37 billion edges and $2^{19}$ vertices, using a combination of provable worst-case query algorithms, accelerated with powerful heuristics.

\textbf{Sparse graph-processing via sketching.}
Graph sketching is not the best solution for graphs that are relatively sparse. 
On sparse graphs, the \sysname approach retains its asymptotic guarantee of low communication but is not space-efficient and may not be able to distribute away its CPU costs. For completeness we evaluate \sysname's performance on sparse graphs and validate these theoretical predictions.
\section{Preliminaries \& Definitions}
\label{sec:preliminaries}
\paragraph{The Graph Streaming Model.}
In the \defn{graph semi-streaming} model~\cite{semistreaming1} (sometimes just called the \defn{graph streaming} model), an algorithm is presented with a \defn{stream} $\graphstream$ of updates (each an edge insertion or deletion) where the length of the stream is $\streamlength$. 
Stream $\graphstream$ defines an input graph $\graph = (\nodes,\edges)$ with $\nodesize = |\nodes|$ and $\edgesize = |\edges|$.  The challenge in this model is to compute (perhaps approximately) some property of $\graph$ given a single pass over $\graphstream$ and at most $O(\nodesize \polylog\nodesize)$ words of memory.  
Each update has the form $((u,v), \Delta)$ where $u,v \in \edges, u \neq v$ and $\Delta \in \{-1,1\}$ where $1$ indicates an edge insertion and $-1$ indicates an edge deletion.  Let $\streamelement_i$ denote the $i$th element of $\graphstream$, and let $\graphstream_i$ denote the first $i$ elements of $\graphstream$.  
Let $\edges_i$ be the edge set defined by $\stream_i$, i.e., those edges which have been inserted and not subsequently deleted by step $i$.  The stream may only insert edge $e$ at time $i$ if $e \notin \edges_{i-1}$, and may only delete edge $e$ at time $i$ if $e \in \edges_{i-1}$. 

In this paper, we consider the following two problems in this model:
\evan{Rephrased this from prior work to us}

\begin{problem}[\textbf{Streaming Connected Components.}]
Given an insert/delete edge stream of length $\streamlength$ that defines a graph $\graph = (\nodes, \edges)$, return a spanning forest of $\graph$.
\end{problem}

\begin{problem}[\textbf{Streaming $k$--Edge Connectivity.}]
Given an insert/delete edge stream of length $\streamlength$ that defines a graph $\graph = (\nodes, \edges)$, for any cut $C \subset \nodes$, return the cardinality of cut $C$ (denoted by $w(C)$) if $w(C) < k$, and return $\infty$ otherwise.
\evan{Switched to $<$ from $\leq$}
\end{problem}

Prior work~\cite{Ahn2012,graphzeppelin} has also considered these problems in the graph streaming model.
\evan{Added this sentence about prior work}

\paragraph{A note on the query model for \sysname.} The graph streaming model above requires computing an answer at the end of the stream. In contrast, \sysname can answer connectivity queries interspersed arbitrarily among insertions and deletions during the stream. It answers both \defn{global connectivity} queries, where the task is to return the connected components or $k$-edge connectivity of the graph, and \defn{batched reachability} queries, where the query consists of a set of vertex pairs $(u_1, v_1), (u_2, v_2)... (u_k, v_k)$ and the task is to determine whether $u_i$ is in the same connected component as $v_i$ for each $i \in [k]$. The goal is to minimize query latency in addition to the goals of minimizing space use and maximizing stream ingestion considered by prior work. For more information see Section~\ref{subsec:queries}.  We assume that the stream is not adaptive, meaning that edge updates and connectivity queries do not depend on the results of prior queries.

\paragraph{Other models in streaming connected components systems.}
Aspen~\cite{aspen} and Terrace~\cite{terrace} are graph-stream processing systems which support connectivity queries. Unlike the semi-streaming model proposed by~\cite{Ahn2012} above, these systems work in the \defn{batch-dynamic model}, where updates are applied to a non-empty graph in batches exclusively containing insertions or deletions. The minimum size of batches is a parameter which affects system performance. Connectivity queries may be issued in between batches, but not during a batch. Aspen is capable of computing queries concurrently with processing updates while Terrace is not. In contrast, \sysname is designed to handle arbitrarily interspersed updates and queries with no notion of batched input. Like Terrace, it does not compute queries concurrently with processing updates.

\section{Sketching Graphs}
\label{sec:sketch}

In this section we briefly review the prior work on connectivity sketching, and then present \sketchname, an improved sketching subroutine that reduces the update cost from $O(\log^2\nodesize)$ to $O(\log \nodesize)$, and reduces the sketch size by a significant constant factor ($73$\%).

\subsection{Prior Work}
Ahn \etal~\cite{Ahn2012} initiate the field of graph sketching with their connected components sketch, which solves the streaming connected components problem in $O(\nodesize \log^3\nodesize)$ space. A key subproblem in their algorithm is \defn{$\ell_0$-sampling}: a vector $x$ of length $n$ is defined by an input stream of updates of the form $(i,\Delta)$ where value $\Delta$ is added to $x_i$, and the task is to sample a nonzero element of $x$ using $o(n)$ space. They use an $\ell_0$-sampler (also called an $\ell_0$-sketch) due to Cormode \etal:

\begin{theorem} (Adapted from ~\cite{l0sketch}, Theorem 1):
\label{thm:l0}
Given a 2-wise independent hash family $\mathcal{F}$ and an input vector $x \in \mathbb{Z}^n$, there is an $\ell_0$-sampler using $O(\log^2(n)\log(1/\delta))$ space that succeeds with probability at least $1 - \delta$.
\end{theorem}

We denote the $\ell_0$ sketch of a vector $x$ as $\sketch(x)$.  The sketch is a linear function, \ie $\sketch(x) + \sketch(y) = \sketch(x+y)$ for any vectors $x$ and $y$. Ahn \etal define a \defn{characteristic vector} $\charvec_u \in \mathbb{Z}_2^{\binom{\nodesize}{2}}$ of each vertex $u \in \nodes$ such that each nonzero element of $\charvec_u$ denotes an edge incident to $u$. 
That is, $\charvec_u \in \mathbb{Z}_2^{\binom{\nodesize}{2}}$ s.t. for all vertices $0 \leq i < j < \nodesize$: 
$$\charvec_u[(i,j)] = \left\{ \begin{array}{ll}
            1 & \quad u \in \{i, j\} \text{ and }(i, j) \in \edges \\
            0 & \quad \text{otherwise} 
        \end{array}\right\}
$$
Further, for any $S \subset \nodes$, the nonzero elements of $\charvec_S = \sum_{u \in S}\charvec_u$ are precisely the set of edges crossing the cut $S, \nodes \setminus S$. For each $u \in \nodes$, the algorithm computes $\sketch(\charvec_u)$: for each edge update $(u,v,\Delta)$ it computes $\sketch(e_{(u,v)})$ where $e_{(u,v)}$ denotes the vector with $1$ in position $\text{idx} = (u, v)$ and 0 in all other positions. It maintains $\sketch(\charvec_u) = \sum_i \sketch(x_i)$ for each $u$. Then for arbitrary $X \subset \nodes$ they can sample an edge across the cut $X, \nodes \setminus X$ by computing $\sketch(\charvec_X) = \sum_{u\in X} \sketch(\charvec_u)$. 

This allows them to perform \Boruvka's algorithm using the sketches: they form $O(\log\nodesize)$ $\ell_0$-sketches $\sketch_0(\charvec_u), \sketch_1(\charvec_u), \cdots , \sketch_{O(\log\nodesize)}(\charvec_u)$ for each $u$. We call $\sketch(\charvec_u) = \bigcup_{i \in [O(\log\nodesize)]}\sketch_i(\charvec_u)$ the \defn{vertex sketch} of $u$, and it has size $O(\log^3 \nodesize)$. Then $\forall u \in 
\nodes$ they query $\sketch_1(\charvec_u)$ to sample an edge incident to $u$. For each resulting component $X$ they compute $\sketch_2(\charvec_X)$ and repeat until all connected components are found. Since each vertex sketch $\sketch(\charvec_u)$ has size $O(\log^3\nodesize)$ bits, the entire data structure has size $O(\nodesize\log^3\nodesize)$. See Appendix~\ref{app:ccsketch} for a more complete description of this algorithm, along with an illustrative example.

\paragraph{Testing k-connectivity.}
Ahn et al.~\cite{Ahn2012} also show how to test k-connectivity of a graph $\graph = (\nodes, \edges)$ (that is, to exactly compute the minimum cut of $\graph$ provided this value is $\leq k$) by constructing a k-connectivity certificate $H = \bigcup_{i \in [k]} F_i $ where $F_0, F_1, \dots, F_{k-1}$ are edge-disjoint spanning forests of $\graph$. $H$ has the property that it is $k'$-edge connected iff $G$ is $k'$-edge connected for all $k' \leq k$. They find each $F_i$ by computing $k$ connectivity sketches of $\graph$ in a single pass over the stream. After the stream, the first connectivity sketch is used to find $F_0$ and the edges of $F_0$ are deleted from the remaining $k-1$ connectivity sketches. The second connectivity sketch can now be used to get $F_1$, whose edges are subsequently deleted from the remaining $k-2$ sketches and so on. The size of the sketches is $O(k\nodesize\log^3\nodesize)$ bits.

\paragraph{\graphzep.}
Tench \etal~\cite{graphzeppelin} present \graphzep, the first implementation of Ahn \etal's connected components algorithm which uses a faster $\ell_0$-sketch algorithm which they call \cubesketch. \graphzep also uses an external-memory-optimized data structure called a \defn{gutter tree} to I/O-efficiently collect updates to be processed, allowing the algorithm to run quickly even when the sketches are stored on disk.



To achieve a failure probability of $\delta$, \graphzep's \cubesketch uses $O(\log^2(n)\log(1/\delta)$ bits of space and has worst-case update time $O(\log(n)\log(1/\delta))$. As in Ahn \etal, they use $O(\log\nodesize)$ of these sketches for each vertex and set $\delta$ to be a small constant. So the space cost per vertex is $O(\log^3\nodesize)$ bits and the worst-case update time is $O(\log^2\nodesize)$. See Appendix~\ref{app:sketch} for a detailed description of \cubesketch.

\subsection{\sysname's new sketch: \sketchname}
\label{subsec:cameo}
We develop a new $\ell_0$ sampler called \sketchname for use in \sysname. \sketchname improves upon \cubesketch with a new update procedure that is a $O(\log\nodesize)$ factor faster to update and reduces space usage by a constant factor via a refined analysis. All other details, including the query procedure, remain unchanged. We present the full details of \sketchname in Appendix ~\ref{app:sketch} and report the asymptotic performance here.

\paragraph{Update procedure.} \sketchname uses a simpler and faster update procedure than \cubesketch.

\begin{theorem}
    \label{thm:cameo_update}
    \cameoupdatethm
\end{theorem}

Theorem~\ref{thm:cameo_update} demonstrates that \sketchname reduces the CPU burden of performing updates and supports Claim~\ref{claim:all}.\ref{itm:cpu}. The proof of this theorem can be found in Appendix~\ref{app:sketch}.

\paragraph{Reduced constant factors.} In \graphzep, Tench \etal use $56 \log (1/\delta)\log n$ bytes of space to guarantee a failure probability of at most $\delta$ when sketching a vector of length $n < 2^{64}$ using \cubesketch. Via a careful constant-factor analysis, we can show that \sketchname can match this failure probability with significantly less space: 


\begin{theorem}
    \label{thm:cameo_cols}
    \cameospacethm
\end{theorem}

Theorem~\ref{thm:cameo_cols} immediately implies a space savings of up to 90\% compared to \cubesketch and thus supports Claim~\ref{claim:all}.\ref{itm:space}. In our implementation, we conservatively choose to use slightly more space than this theorem requires to reduce the failure probability further. Still, our implementation requires only $2/7$ths of the space used in \graphzep~\cite{graphzeppelin} (see Section~\ref{sec:system} for details).

See Appendix~\ref{app:cameo-proofs} for the proof of this theorem.

\section{\sysname Design}
\label{sec:bandwidth}
\begin{figure}
    \centering
    \includegraphics[width=.45\textwidth]{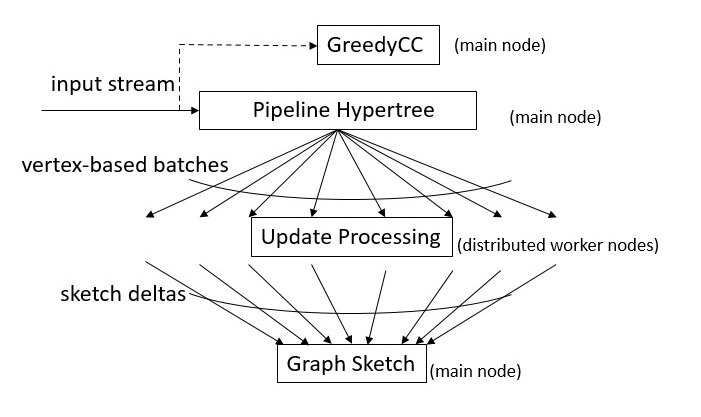}
    \caption{Data flow diagram for \sysname's ingestion algorithm. Connectivity information from the input stream is compressed into the graph sketch and added to \dsuname.}
    \label{fig:sys_diagram}
\end{figure}

\sysname uses \sketchnames to compute connectivity. The CPU work of computing updates to $\sketch(\graph)$ is done by distributed workers, while $\sketch(\graph)$ itself is stored on the main node. To answer queries, the main node computes a spanning forest via \Boruvka's algorithm using $\sketch(\graph)$ as described in Section~\ref{sec:sketch} or via a heuristic algorithm which we call \dsuname and describe in Appendix~\ref{subsec:dsu}.

In this section we describe how \sysname's main node efficiently collects updates into vertex-based batches to minimize communication, how the distributed workers process these batches of updates into sketch form, and how \sysname answers connectivity queries. We also prove the asymptotic upper bounds on the CPU and communication costs incurred by these operations. This analysis explains \sysname's surprising performance profile. 


Figure~\ref{fig:sys_diagram} summarizes the \sysname data flow. The input stream, consisting of edge insertions, edge deletions, and queries arrives at the main node. Updates (insertions and deletions) are inserted into the \defn{\treename} (Section~\ref{subsec:pht}) and also into \dsuname (Section~\ref{subsec:dsu}). The \treename collects updates into \defn{vertex-based batches} (Section~\ref{subsec:batch}) which are sent to distributed worker nodes. These worker nodes process batches by running the \sketchname algorithm, producing \defn{sketch deltas} (Section~\ref{subsec:dist_sketch}) which are applied to the \sketchnames on the main node. 


\subsection{Ingesting Stream Updates on the Main Node}

\subsubsection{Vertex-Based Batching}
\label{subsec:batch}

The core technique that makes distributed sketch processing communication-efficient (and therefore feasible) is \emph{vertex-based batching}, where many updates with a common endpoint are collected into a batch. Intuitively, because one or many updates to the same endpoint can be represented as a single sketch delta of fixed size, batching updates by endpoint drastically reduces the communication cost of sending sketch deltas from worker nodes to the main node.

Specifically, any updates for $u\in \nodes$ are collected into a batch $B_u \subseteq \{(x,y) \in \edges \mid x = u \vee y = u\}$. $B_u$ is sent to a single distributed worker, which returns a sketch of the updates. As we will see in Section~\ref{subsec:dist_sketch} this sketch has size $\sketchsize = O(\log^3\nodesize)$ bits. During update processing, \sysname only sends $B_u$ when $|B_u| \geq \alpha \sketchsize / \log \nodesize$ for some constant $\alpha\geq1$. Each update requires $\log V$ bits to represent, so a buffer that contains $\alpha \sketchsize / \log \nodesize$ updates has size $\alpha \sketchsize$ bits.

As a result of this policy, the amortized communication cost per update is small. Say the stream contains $\streamlength$ updates. The bandwidth cost to receive the input stream is $\streamlength$. Since each update is included in two vertex-based batches (one per endpoint) and each batch is sent to a distributed worker once, the total bandwidth cost of sending vertex-based batches is $2\streamlength$. Finally, each vertex-based batch induces the distributed worker that receives it to respond with a sketch delta (which is $1/\alpha$ of the batch's size). This means that as long as \sysname is processing full vertex-based batches, the network bandwidth cost of processing $\streamlength$ updates is at most $(3+1/\alpha)\streamlength$. This technique is simple, but crucial for good performance.

\subsubsection{Pipeline Hypertree}
\label{subsec:pht}

To make vertex-based batching fast, we design the \defn{\treename}, which is a simplified and parallel variant of the buffer tree~\cite{Arge03} designed to minimize cache line misses and thread contention. The \treename receives arbitrarily ordered stream updates and consolidates them into vertex-based batches. Each update inserted into the \treename is moved $O(\log_{\cachesize/\cachelinesize}\nodesize)$ times before being returned in a vertex-based batch, where $\cachesize$ denotes the size of L3 cache and $\cachelinesize$ denotes the size of an L3 cache line.  The total size of the data structure is $O(\nodesize\log^3\nodesize)$ bits. We defer description of the design and implementation of the \treename to Appendix~\ref{app:pht}.

\subsection{Distributed Sketch Processing \hfill}
\label{subsec:dist_sketch}

\sketchnames have strong data locality, which we exploit for parallelism: processing a graph edge update $(u,v,\Delta)$ requires updating only $\sketch(\charvec_u)$ and $\sketch(\charvec_v)$ and these updates can be performed independently of each other. Because the sketches are linear, the sketch update for $(u,v)$ can be computed on its own and later summed to the sketch of $u$. This means that the vast majority of the computation required to process $(u,v)$ can be performed before accessing the sketches for $u$ and $v$. \sysname exploits this independence to distribute the computational cost of these updates while storing the sketches on a single worker.

When a batch of updates $(e_1, e_2, ...)$ for vertex $u$ is sent to a worker node, the worker node computes $\sum_{j}\sketch(e_j)$, which we call a \defn{sketch delta} (since it is a sketch encoding the change in the neighborhood of vertex $u$. Note that it has size $O(\log^3\nodesize)$---equal to the size of a vertex sketch. 
This immediately gives the following result:

\begin{theorem}[Distributed cost.]
The distributed CPU cost of processing a batch of $x$ updates for vertex $u \in \nodes$ into a sketch delta $\sketch_u$ is $O(x\log(\nodesize))$.
\end{theorem}

\paragraph{Sketch merging.}
\sysname's main node maintains the \defn{graph sketch}: $\sketch(\graph) = \bigcup_{u \in \nodes}\sketch(\charvec_u)$.
After a sketch delta for vertex $u$ is created by a distributed worker is then sent to the main node where it is added to $\sketch(\charvec_u)$.
The graph sketch is stored in RAM which is feasible since it has total size $O(\nodesize \log^3(\nodesize))$ bits.

\subsection{Finding Connected Components}
\label{subsec:queries}
\evan{Changed section name and some wording.}
\sysname is designed to efficiently answer two types of connected component queries: \defn{global connectivity queries} where the task is to map each vertex to its connected component, and \defn{batched reachability queries} where the query consists of a set of vertex pairs $(u_1, v_1), (u_2, v_2)... (u_k, v_k)$ and the task is to determine whether $u_i$ is in the same connected component as $v_i$ for each $i \in [k]$. At a high level, \sysname answers these queries by producing a spanning forest of the graph defined by the input stream. This is done as described in Section ~\ref{sec:sketch} via Boruvka's algorithm on the vertex sketches.

Processing queries can increase network bandwidth costs if not handled carefully. Computing the spanning forest from the sketches can only be done after all pending stream updates have been processed. We say the graph sketch stored on the main node is \defn{current} with respect to query $q$ at time $t$ if there are no pending updates; \ie all updates that arrived prior to time $t$ have been processed and merged into the graph sketch. If for some vertex $i$ there are a small number of updates for $i$ in the \treename, then by the reasoning in Section ~\ref{subsec:batch} this could incur an average communication cost per update of $O(\log^2\nodesize)$. In the worst case, the input stream could insert a perfect matching into an empty graph and then issue a query; the average communication cost per update would be $O(\log^2\nodesize)$.

\sysname avoids this problem by adopting a hybrid distribution policy for pending updates when a query is issued.  When a query is issued, \sysname first flushes the \treename so all pending updates are stored in the leaves. For each leaf, if the leaf is at least a $\gamma$-fraction full, it is sent as a vertex-based batch to a distributed worker, where $\gamma \in (0,\frac{1}{2}]$ is a parameter chosen by the user. All leaves that are less than a $\gamma$-fraction full are processed locally on the main node, costing no additional network bandwidth.

As a consequence of this policy, we have the following theorems. 

\begin{restatable}[Communication cost.]{theorem}{commthm}
\label{thm:queries}
The communication cost of ingesting $\streamlength$ updates and answering $Q$ queries is at most $(3 + 1/(\gamma\alpha)) \streamlength$, where $\gamma$ and $\alpha$ are constants.
\end{restatable}
\begin{proof}
See Appendix~\ref{app:proofs}.
\end{proof}



Importantly, this means that \sysname never uses more than a constant multiple of the network bandwidth required to receive the input stream, \emph{regardless of the number or distribution of queries.}

\begin{restatable}[Computational cost on main node.]{theorem}{compthm}
\label{thm:comp}
Given a input stream of length $\streamlength$ defining $\graph = (\nodes, \edges)$ and a series of $Q$ connectivity queries issued throughout the stream, let $\streamlength_i$ denote the number of edge updates that arrive after query $Q_i$ but before query $Q_{i+1}$. \sysname never uses more than $O(\nodesize\log^3\nodesize)$ bits of space on the main node while processing the stream, and the amortized cost per update to process $\streamlength_i$ is $O(\log_{\cachesize/\cachelinesize}(\nodesize))$ if $\streamlength_i = \Omega(\nodesize \log^2(\nodesize))$ and $O(\log(\nodesize))$ otherwise. Computing each query $Q_i$ takes $O(\nodesize\log^2(\nodesize))$ time.
\evan{Corrected some claims in this theorem}
\end{restatable}

\begin{proof}
See Appendix~\ref{app:proofs}.
\end{proof}

As a result, when $\streamlength_i = \Omega(\nodesize \log^2\nodesize)$ the amortized CPU cost for all computation on the main node is $O(\log_{\cachesize/\cachelinesize}(\nodesize))$ (the amortized cost to process updates with the \treename). For reasonable values of $\cachesize$ and $\cachelinesize$ this logarithm evaluates to a small constant, typically 3. Moreover, these operations are just data movement operations, so they can be done at near RAM bandwidth.



\subsection{Computing k-connectivity}
\sysname's architecture can in principle accomodate many other graph sketch algorithms that use connectivity as a subroutine. For instance, since the k-connectivity sketch requires maintaining $k$ independent copies of the connectivity sketch, we can achieve comparable computation and communication upper bounds by slightly modifying the procedure used to distribute connectivity sketching. Set the size of a vertex-based batch and the size of leaf node buffers in the \treename to $\alpha\cdot k \log^3\nodesize$ (the per-vertex sketch size for k-connectivity). When a distributed worker recieves a vertex-based batch, it computes the sketch delta of the batch for all $k$ copies of the connectivity sketch and sends this back to the main node. To answer k-connectivity queries, \sysname produces a k-connectivity certificate using the query algorithm summarized in Section ~\ref{sec:sketch}. This immediately gives the following result:

\begin{theorem}[k-connectivity.]
\label{thm:kconn}
Given a input stream of length $\streamlength$ defining $\graph = (\nodes, \edges)$ and a series of $Q$ connectivity queries issued throughout the stream, let $\streamlength_i$ denote the number of edge updates that arrive after query $Q_i$ but before query $Q_{i+1}$. 
\sysname never uses more than $O(k\nodesize\log^3\nodesize)$ bits of space on the main node while processing the stream.
The main node amortized cost per update to process $\streamlength_i$ is $O(\log_{\cachesize/\cachelinesize}(\nodesize))$ if $\streamlength_i = \Omega(k\nodesize \log^2\nodesize)$ and $O(k\log\nodesize)$ otherwise. 
Computing each query $Q_i$ takes $O(k^2\nodesize\log^2\nodesize)$ time.
The distributed CPU cost of processing a batch of $x$ updates for vertex $u \in \nodes$ into a sketch delta $\sketch_u$ is $O(xk\log\nodesize)$.
The communication cost of ingesting $\streamlength$ updates and answering $Q$ queries is at most $(3+1/(\gamma\alpha))\streamlength$, where $\gamma$ and $\alpha$ are constants.
\evan{Corrected some claims in this theorem.}
\end{theorem}

Note that the network communication cost does not increase above that of connectivity, and for sufficiently infrequent queries the cost to the main node is also independent of $k$. See Section~\ref{subsec:k_conn_partial} for experimental confirmation of these results.


\section{\sysname Implementation}
\label{sec:system}

Processing stream updates into the graph sketch is a computationally intensive process: for a moderately sized data-set with $2^{18}$ vertices, applying a single edge update requires evaluating 184 hash functions. \sysname farms out this computationally intensive portion of the workload to worker nodes while the other portions of stream ingestion, including update buffering and sketch storage, remain the responsibility of the main node. \sysname uses $164\nodesize*(\log^2\nodesize-\log\nodesize)$B of space on the main node to store the sketches and the \treename. On the worker nodes, \sysname requires storage for a single sketch and batch per CPU. Thus, a worker node with $t$ threads requires $t \cdot 164(\log^2\nodesize-\log\nodesize)$ bytes. On a billion-vertex graph, each worker thread requires only 64 KiB of RAM. When computing $k$-connectivity, all of the above costs are multiplied by a factor $k$.

\sysname must handle two tasks: stream ingestion, where edge updates from the input stream are compressed into the graph sketch; and query processing, where connectivity queries are computed from the graph sketch (or sometimes from auxiliary query-accelerating data structures, described below). We defer a full description of the implementation, including auxiliary data structures, parameter choices, and software tools, to Appendix~\ref{app:system}.


\section{Experiments}
\label{sec:experiments}

\paragraph{Experiment setup.} We implemented \sysname in C++14 and compiled using g++ version 9.3 with openmpi 4.1.3 for Linux.
We ran our experiments on an AWS cluster composed of an c5n.18xlarge instance for a main node and 40 c5.4xlarge instances as worker nodes. These instances have respectively 36 and 8 2-way hyperthreaded Intel(R) Xeon(R) Platinum 8124M CPU @ 3.00GHz cores with respectively 196 GB and 32 GB of RAM.

Each of our worker nodes only requires 2 GB of RAM because sketch deltas are small and workers are stateless. However, AWS workers with sufficient CPU power come with more RAM than we need.


\paragraph{A note about experimental comparisons.}
Ideally, we would include experimental comparisons against existing distributed systems that solve connectivity or k-connectivity on graph streams with edge insertions and deletions. However, only one such system (KickStarter~\cite{kickstarter}) exists in the literature, and its source code is not available (see Appendix~\ref{sec:related}). Instead, we compare against a theoretical upper bound for stream ingestion: the data acquisition cost. 

We define the \defn{data acquisition cost} to be the cost of receiving the input stream at the main node over a network link and logging it in RAM. Sequential RAM bandwidth is a trivial bound on this cost.
\begin{itemize}[leftmargin=*]
\item \defn{RAM sequential bandwidth} is the maximum rate at which the CPU can write consecutive words in RAM on any number of threads. 

\item \defn{RAM random access bandwidth} is the maximum worst-case rate at which the CPU can write any sequence of words in RAM on any number of threads. 



\end{itemize}


\paragraph{Experiment Metrics.} Our experiments process the entire stream of updates and then perform a single query at the end of the stream. Where specified experiments also perform additional queries throughout the processing of the stream.

We report the update throughput by measuring wall-clock time from the beginning of the stream until all updates have been applied to the sketches. We also measure the total amount of network communication to/from the main node and the amount of RAM usage on the main node.

When performing a query we measure the wall-clock latency from the moment the query is issued to the time the answer is returned to the user. This time includes the latency of flushing all pending updates from the \treename and then the query computation itself.

\subsection{Datasets}
\begin{table}
\begin{center}
\def\arraystretch{1.1}
\footnotesize
\caption{Datasets used in our experiments.}
\label{tab:datasets}
\begin{tabular}{ |c|c|c|c| } 
 \hline
 Name & Vertices & Edges & Stream Updates\\ 
 \hline
 \textbf{kron13} & $2^{13}$ & $1.7\times 10^7$ & $1.2\times 10^8$\\
 \hline
 \textbf{kron15} & $2^{15}$ & $2.7 \times 10^8$ & $1.9 \times 10^9$\\
 \hline
 \textbf{kron16} & $2^{16}$ & $1.1 \times 10^9$ & $7.7\times 10^9$\\
 \hline
 \textbf{kron17} & $2^{17}$ & $4.3 \times 10^9$ & $3.1\times 10^{10}$\\
 \hline
 \textbf{ca-citeseer} & $2.3 \times 10^6$ & $8.1 \times 10^5$ & $1.1 \times 10^8$ \\ 
 \hline
 \textbf{p2p-gnutella} & $6.3\times 10^4$ & $1.5\times 10^5$ & $1.9 \times 10^6$ \\ 
 \hline
 \textbf{rec-amazon} & $9.2\times 10^4$ & $1.3\times 10^5$ & $1.7 \times 10^6$ \\ 
 \hline
 \textbf{google-plus} & $1.1 \times 10^5$ & $1.4 \times 10^7$ & $1.9\times 10^8$ \\ 
 \hline
 \textbf{web-uk-2005} & $1.3 \times 10^6$ & $1.2 \times 10^8$ & $1.6 \times 10^9$ \\ 
 \hline
 \textbf{erdos18} & $2^{18}$ & $1.7 \times 10^{10}$ & $4 \times 10^{10}$ \\
 \hline
 \textbf{erdos19} & $2^{19}$ & $3.4 \times 10^{10}$ & $4 \times 10^{10}$\\
 \hline
 \textbf{erdos20} & $2^{20}$ & $8 \times 10^{10}$ & $1 \times 10^{11}$\\
 \hline
 
\end{tabular}
\end{center}
\end{table}

In many of the experiments below, we use the synthetically generated graph streams used in the evaluation of \graphzep~\cite{graphzeppelin}. These graphs were generated using the Graph500 Kronecker graph generator specification~\cite{Ang2010IntroducingTG}, and are very dense: each graph contains approximately $1/4$ of all possible edges. 

For larger-scale experiments, we evaluate \sysname on randomly generated Erdos-Renyi graphs (with $2^{18}$, $2^{19}$, and $2^{20}$ vertices) with edge probability set to $1/4$.

Finally, we evaluate \sysname on real-world graph datasets from the SNAP graph repository~\cite{snapnets} and NetworkRepository~\cite{netrepo}.

All of the above graphs were transformed into a random streams of edge insertions and deletions using the method described in~\cite{graphzeppelin}. We additionally inserted and removed all edges seven times to increase stream length. Each random stream, once insert/delete pairs for the same edge are removed, is exactly the edge list of the graph used to generate it.
These datasets are summarized in Table~\ref{tab:datasets}.


\subsection{\sysname is Highly Scalable.}
\label{subsec:scaling}
\begin{figure}
    \centering
    \includegraphics[width=.45\textwidth]{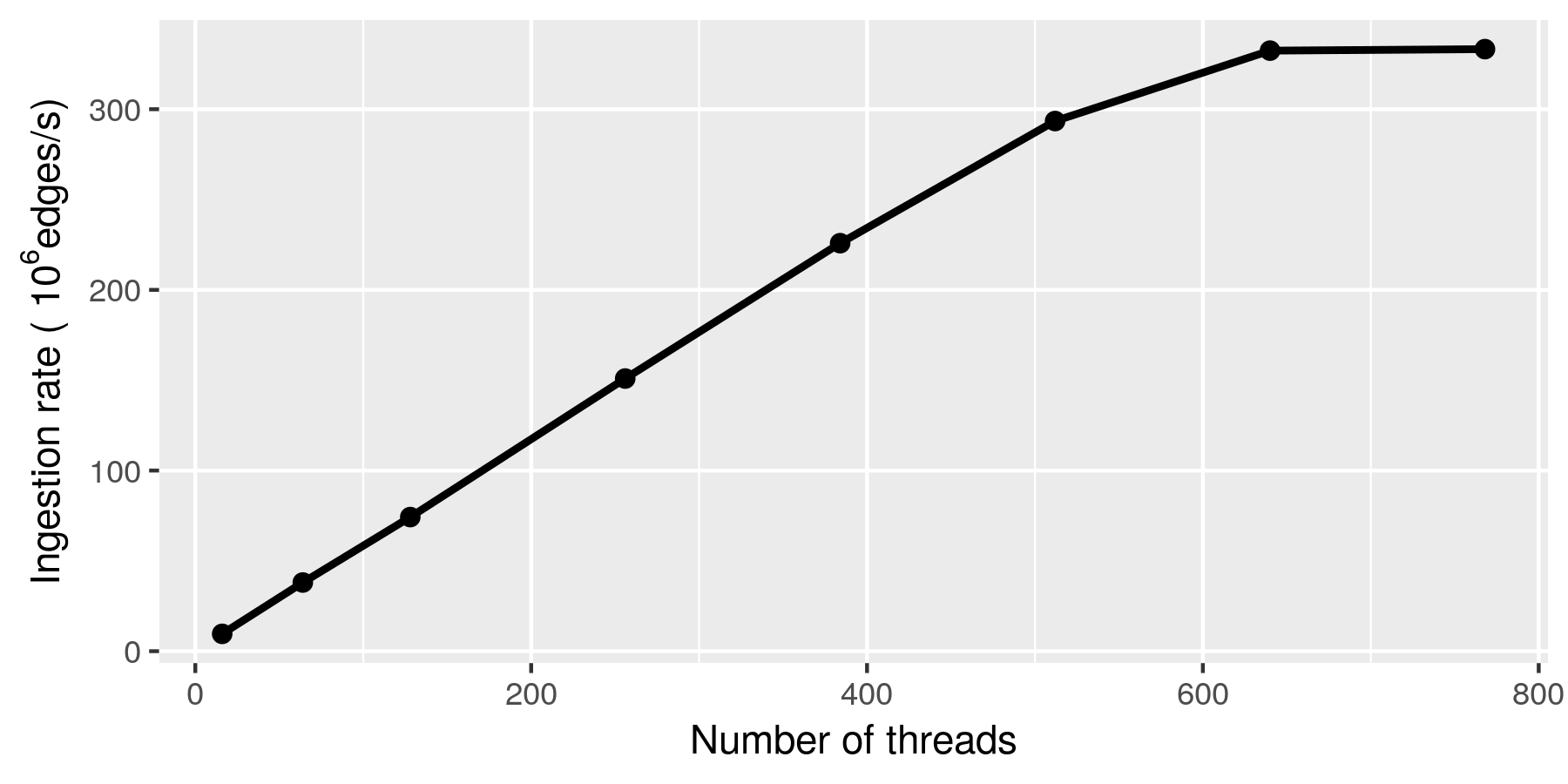}
    \caption{\sysname ingestion rate scales to one-fourth of sequential RAM bandwidth.}
    \label{fig:scaling}
\end{figure}


We measured \sysname's stream ingestion rate and network bandwidth usage on kron17 given varying numbers of distributed workers. 
Figure~\ref{fig:scaling} demonstrates a near-linear increase in ingestion rate as more distributed threads are added, until throughput levels off as it approaches 340 million updates/sec. The ingestion rate on one worker node (with 16 threads) is about 9.6 million updates/sec, and the ingestion rate on 40 worker nodes (with a total of $40 \cdot 16 = 640$ threads) is 332 million updates/sec, a $35\times$ speedup. We note that at this point \sysname is observably \defn{not} CPU-bound: adding more worker nodes no longer increases throughput, and we measure instructions per cycle on the main node CPU to be 0.8, indicating that the CPU is not instruction bound, but RAM bandwidth bound~\cite{gregg_2017}.


Since graph stream updates are 9 bytes, \sysname ingestion bandwidth ($2.8$ GiB/sec) \evan{Added number} is four times that of random-access RAM bandwidth ($759$ MiB/sec), and roughly one-fourth of sequential RAM access bandwidth ($12.4$ GiB/sec) on the (c5n.18xlarge) main node.

Throughout these tests we observe a constant amount of network communication used; approximately $1.6$ times the size of the input stream for dense graphs. See Table~\ref{tab:kconnect_full} for a more complete evaluation.

\begin{figure}
    \centering
    \includegraphics[width=.45\textwidth]{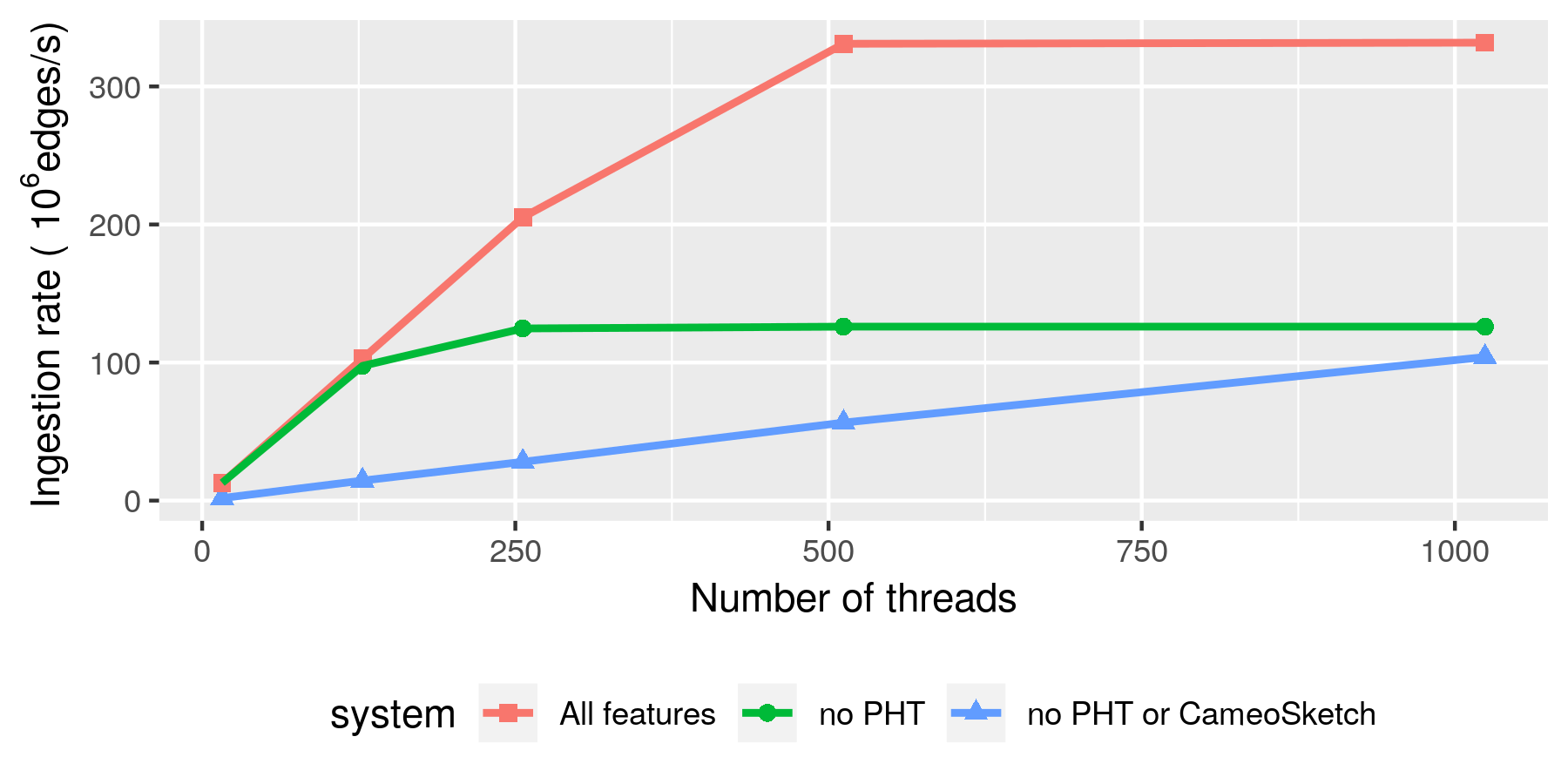}
    \caption{\sketchname and \treename are vital for good ingestion performance. Without \sketchname, \sysname's ingestion rate scales slowly as the number of threads increases. Without \treename, the system bottlenecks at slightly over 100 million updates/sec.}
    \label{fig:ablative}
\end{figure}
\paragraph{Performance impact of \sketchname and \treename.}
Figure~\ref{fig:ablative} illustrates the performance impact of \sketchname and the \treename on \sysname's ingestion rate with the kron17 dataset. When \sysname uses \graphzep's buffering data structure and its \cubesketch sketch algorithm, its ingestion rate increases as more distributed workers are added, but at a slow rate. Using \sketchname alongside \graphzep's buffering system results in a much faster increase in ingestion rate, but the system bottlenecks at roughly 120 million updates/sec. In contrast, the full \sysname system continues to increase its ingestion rate dramatically as more workers are added, and bottlenecks at over 300 million updates/sec. We conclude that the $O(\log\nodesize)$ decrease in ingestion cost for \sketchname and the improved design of the \treename are vital for \sysname's performance. See Appendix~\ref{subsec:gzcomp} for a direct experimental comparison of \sysname to \graphzep on a single machine.

\begin{table}
\begin{center}
\def\arraystretch{1.1}
\footnotesize
\caption{\sysname has a high ingestion rate on sufficiently dense graphs and has low communication overhead.}
\label{tab:speeds}
\begin{tabular}{ |c|c|c|c| } 
 \hline
 Dataset & Ingestion Rate & Communication (as a\\
  & ($\times 10^6$ updates/sec) & factor of stream size)\\ 
 \hline
 \textbf{kron13} & $231$ & 1.6\\
 \hline
 \textbf{kron15} & $336$ & 1.6\\ 
 \hline
 \textbf{kron16} & $334$ & 1.6\\ 
 \hline
 \textbf{kron17} & $335$ & 1.6\\ 
 \hline
 \textbf{ca-citeseer} & $17.3$ & 1.7\\ 
 \hline
 \textbf{p2p-gnutella} & $13.5$ & 0\\ 
 \hline
 \textbf{rec-amazon} & $12.5$ & 0\\ 
 \hline
 \textbf{google-plus} & $134$ & 2.6\\ 
 \hline
 \textbf{web-uk-2005} & $91.5$ & 3.4\\ 
 \hline
 \textbf{erdos18} & 291 & 1.6\\
 \hline
 \textbf{erdos19} & 226 & 1.6\\
 \hline
 \textbf{erdos20} & 236 & 1.6\\
 \hline
 
\end{tabular}
\end{center}
\end{table}

\paragraph{More datasets.} Table~\ref{tab:speeds} summarizes \sysname's stream ingestion rate and network costs using 640 worker threads on a variety of synthetic and real-world datasets. 
Its ingestion rate is very high on dense graph streams and on real-wold streams google-plus and web-uk-2005. It is lower on ca-citeseer, p2p-gnutella and rec-amazon because these datasets do not contain enough stream updates to surpass \sysname's $4\%$ leaf fullness threshold. So, for these streams, a large portion (or all) of update processing occurs on the main node. 

\paragraph{Circumventing bottlenecks.} These experiments support our claims that sketching can avoid the traditional bottlenecks in distributed graph stream processing. They demonstrate that CPU cost can be distributed away, supporting Claim~\ref{claim:all}.\ref{itm:cpu};
that network communication can be only a constant factor of the data acquisition cost, supporting Claim~\ref{claim:all}.\ref{itm:comm};
and that stream processing bandwidth can reach to within a factor 4 of sequential RAM bandwidth, supporting Claim~\ref{claim:all}.\ref{itm:ram}.

\subsection{\sysname Answers Queries Quickly.}
\label{subsec:queryex}
We measured \sysname's query latency for both global connectivity queries and batched reachability queries on an extended kron17 stream.
For batched reachability queries, we uniformly sample the pairs to query $(v_1, v_2), (v_3, v_4), \dots$ from the set of all vertices. We issue queries periodically and record (1) the time to flush the \treename and update the graph sketch and (2) to perform \Boruvka's algorithm using the graph sketch. The sum of these times is the total query latency, though only the \Boruvka computation is additional work induced by the query; the flushing work would have happened anyway as part of stream ingestion even if the query was never issued. 
We found that flushing takes roughly $2.3$ seconds, while \Boruvka's algorithm takes $0.3$ seconds. 

\begin{figure}
    \centering
    \includegraphics[width=.45\textwidth]{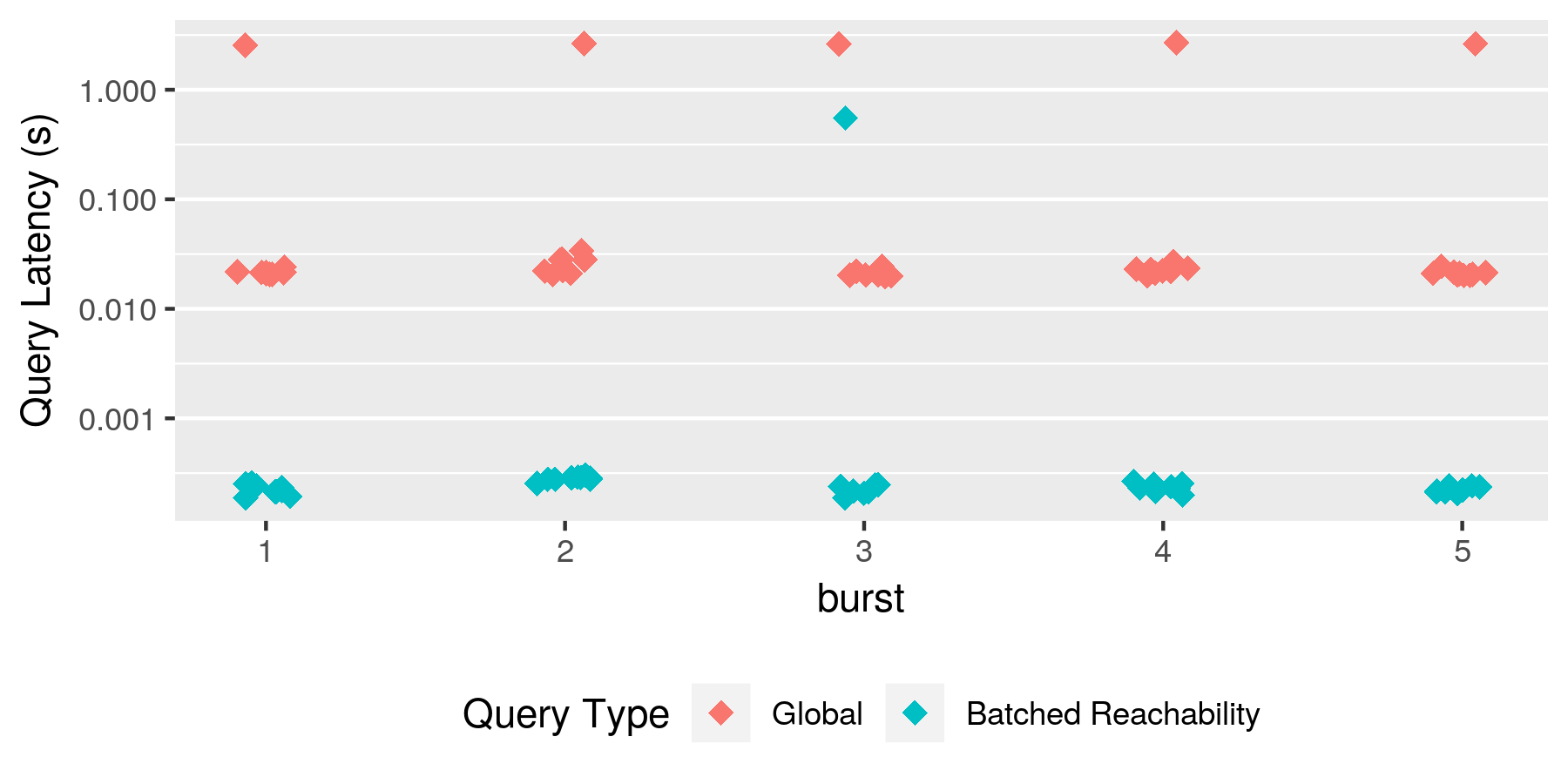}
    \caption{\dsuname dramatically decreases query latency.}
    \label{fig:dsu_query}
\end{figure}
We also evaluate how query latency and ingestion rate are affected by the use of \dsuname. 
We measure the performance impact of this optimization by running \sysname on extended kron17 and issuing \defn{query bursts}: multiple connectivity queries issued close together (within 50 thousand stream updates). We track the latency of each query in the burst. The results are summarized in Figure~\ref{fig:dsu_query}. The first query in each burst, which requires flushing and running \Boruvka's algorithm to answer, has high (multi-second) latency. However, we see that for the remaining queries in the burst the latency is much lower: two orders of magnitude lower for global connectivity queries, and up to four orders of magnitude lower for batched reachability queries. 


\subsection{k-connectivity Performance}
\label{subsec:k_conn_partial}
\begin{table}
\begin{center}
\def\arraystretch{1.1}
\footnotesize
\caption{When computing k-connectivity, increasing $k$ incurs a linear decrease in ingestion rate, a linear increase in sketch size, a quadratic increase in query latency, and does not affect total network communication. Reported values obtained by running \sysname on the kron17 dataset.}
\label{tab:kconnect}
\begin{tabular}{ |c|c|c|c|c| } 
 \hline
 & Ingestion Rate  & Sketch Size & Query & Network \\ 
 & ($10^6$ u/s) & (GiB) & (seconds) & (GiB) \\
 \hline
 \textbf{$k=1$} & $338.5$ & $15.25$ & $1.27$ & $425.2$\\
 \hline
 \textbf{$k=2$} & $200$ & $24.40$ & $5.02412$ & $464.981$\\ 
 \hline
 \textbf{$k=4$} & $101.5$ & $46.49$ & $16.1608$ & $468.886$\\ 
 \hline
 \textbf{$k=8$} & $50.76$ & $83.39$ & $65.5716$ & $476.598$\\ 
 \hline
 
\end{tabular}
\end{center}
\end{table}
We repeat our scaling, network communication, and query latency experiments on the $k$-connectivity problem for a variety of datasets and values of $k.$ Table ~\ref{tab:kconnect} highlights \sysname's performance for computing $k$-connectivity on the kron17 dataset. Note that increasing $k$ incurs a linear decrease in ingestion rate, a linear increase in sketch size, a quadratic increase in query latency, and has no significant effect on network communication. These experimental results support the asymptotic conclusions in Theorem ~\ref{thm:kconn}.
See Appendix~\ref{app:kconnect_full} for a summary of \sysname's performance on more datasets.

\section{Conclusion}


This paper demonstrates how to use linear sketching to avoid the space, CPU \& network bottlenecks faced by existing graph processing systems. We make our argument in the context of the \sysname system for finding connected components of dense, dynamic graphs.

By avoiding these bottlenecks, \sysname achieves remarkable performance. Specifically, it supports a stream ingestion rate one-fourth of sequential RAM bandwidth.
\iftrue 
\section{Appendix}
\evan{I added this little blurb}
The text of the appendix can be found in the full version of this paper\footnote{Full paper: \protect\url{https://arxiv.org/abs/2410.07518}}.
\fi

\section*{Acknowledgments}
This work was supported in part by NSF grants CCF-2247577, 
CCF-2106827, 
and NRT-HDR 2125295.


\bibliographystyle{plain}
\bibliography{main}
\appendix

\clearpage

\section{More Detail on Ahn \etal's~\cite{Ahn2012} algorithm}
\label{app:ccsketch}

Recall the definition of a \defn{characteristic vector} from Section~\ref{sec:sketch}: $\charvec_u \in \mathbb{Z}_2^{\binom{V}{2}}$ of each vertex $u \in \nodes$ such that each nonzero element of $\charvec_u$ denotes an edge incident to $u$. 
That is, $\charvec_u \in \mathbb{Z}_2^{\binom{\nodesize}{2}}$ s.t. for all vertices $0 \leq i < j < \nodesize$: 
$$\charvec_u[(i,j)] = \left\{ \begin{array}{ll}
            1 & \quad u \in \{i, j\} \text{ and }(i, j) \in \edges \\
            0 & \quad \text{otherwise} 
        \end{array}\right\}
$$
First, we give an example of how one can represent a graph as a set of characteristic vectors, how these vectors can be updated as edges are inserted or deleted, and how they can be used to compute connected components using \Boruvka's algorithm. Then, we outline how Ahn \etal use linear sketching to reduce the space cost of this approach to $O(\nodesize\log^3\nodesize)$.

Consider the graph $\graph$ its the characteristic vectors
$\{f_i \mid i \in [\nodesize]\}$ for each of its nodes as shown in Figure~\ref{fig:example-graph}.

\begin{figure}
    \centering

    \begin{subfigure}{\linewidth}
    \centering

    \tikzset{every picture/.style={line width=0.75pt}} 
    
    \begin{tikzpicture}[x=0.75pt,y=0.75pt,yscale=-1,xscale=1]
    
    \draw   (6.33,24.52) .. controls (6.33,19.09) and (10.65,14.68) .. (15.97,14.68) .. controls (21.3,14.68) and (25.62,19.09) .. (25.62,24.52) .. controls (25.62,29.95) and (21.3,34.35) .. (15.97,34.35) .. controls (10.65,34.35) and (6.33,29.95) .. (6.33,24.52) -- cycle ;
    
    \draw   (9.4,99.86) .. controls (9.4,94.43) and (13.71,90.03) .. (19.04,90.03) .. controls (24.36,90.03) and (28.68,94.43) .. (28.68,99.86) .. controls (28.68,105.29) and (24.36,109.7) .. (19.04,109.7) .. controls (13.71,109.7) and (9.4,105.29) .. (9.4,99.86) -- cycle ;
    
    \draw   (123.05,15.83) .. controls (123.05,10.4) and (127.37,6) .. (132.69,6) .. controls (138.02,6) and (142.33,10.4) .. (142.33,15.83) .. controls (142.33,21.27) and (138.02,25.67) .. (132.69,25.67) .. controls (127.37,25.67) and (123.05,21.27) .. (123.05,15.83) -- cycle ;
    
    \draw   (73.76,55.17) .. controls (73.76,49.73) and (78.08,45.33) .. (83.4,45.33) .. controls (88.72,45.33) and (93.04,49.73) .. (93.04,55.17) .. controls (93.04,60.6) and (88.72,65) .. (83.4,65) .. controls (78.08,65) and (73.76,60.6) .. (73.76,55.17) -- cycle ;
    
    \draw   (47.96,133.06) .. controls (47.96,127.63) and (52.28,123.23) .. (57.6,123.23) .. controls (62.93,123.23) and (67.25,127.63) .. (67.25,133.06) .. controls (67.25,138.49) and (62.93,142.9) .. (57.6,142.9) .. controls (52.28,142.9) and (47.96,138.49) .. (47.96,133.06) -- cycle ;
    
    \draw [color={rgb, 255:red, 208; green, 2; blue, 27 }  ,draw opacity=1 ][line width=2.25]    (25.62,24.52) -- (73.76,55.17) ;
    \draw [color={rgb, 255:red, 208; green, 2; blue, 27 }  ,draw opacity=1 ][line width=2.25]    (25.62,22.6) -- (123.05,15.83) ;
    \draw [color={rgb, 255:red, 208; green, 2; blue, 27 }  ,draw opacity=1 ][line width=2.25]    (15.97,34.35) -- (19.04,90.03) ;
    \draw [color={rgb, 255:red, 208; green, 2; blue, 27 }  ,draw opacity=1 ][line width=2.25]    (26.13,106.89) -- (49.88,126.55) ;
    \draw [color={rgb, 255:red, 208; green, 2; blue, 27 }  ,draw opacity=1 ][line width=2.25]    (57.6,123.23) -- (77.85,63.46) ;
    
    \draw (10.92,16.31) node [anchor=north west][inner sep=0.75pt]   [align=left] {1};
    \draw (52.55,124.86) node [anchor=north west][inner sep=0.75pt]   [align=left] {5};
    \draw (13.99,91.66) node [anchor=north west][inner sep=0.75pt]   [align=left] {3};
    \draw (78.35,46.96) node [anchor=north west][inner sep=0.75pt]   [align=left] {2};
    \draw (127.64,7.63) node [anchor=north west][inner sep=0.75pt]   [align=left] {4};

    \end{tikzpicture}
    \caption{Graph $\graph$.}
    \end{subfigure}\vspace{2mm}

    \begin{subfigure}{\linewidth}
        \centering
        \begin{tabular}{|c|c c c c c c c c c c|}
        \hline
            ~ & (1,2) & (1,3) & (1,4) & (1,5) & (2, 3) & (2,4) & (2,5) & (3,4) & (3,5) & (4,5) \\ \hline
            1 & 1 & 1 & 0 & 0 & 0 & 0 & 0 & 0 & 0 & 0 \\ \hline
            2 & 1 & 0 & 0 & 0 & 0 & 1 & 1 & 0 & 0 & 0 \\ \hline
            3 & 0 & 1 & 0 & 0 & 0 & 0 & 0 & 0 & 1 & 0 \\ \hline
            4 & 0 & 0 & 0 & 0 & 0 & 0 & 0 & 0 & 0 & 0 \\ \hline
            5 & 0 & 0 & 0 & 0 & 0 & 1 & 1 & 0 & 1 & 0 \\ \hline
        \end{tabular}
        \caption{$\graph$'s characteristic vectors.}
    \end{subfigure}

    \caption{An example graph $\graph$ and its characteristic vectors.}
    \label{fig:example-graph}
\end{figure}

When we receive an edge insertion or deletion for an edge $(u,v)$, we can update the characteristic vectors $f_u$ and $f_v \in \mathbb{Z}_2^{\binom{V}{2}}$ by adding the one-hot vector $e_{(u,v)}$ to each of them. 
Specifically, $e_{(u,v)}$ is a vector with a $1$ in the position corresponding to the edge $(u,v)$ and $0$ in all other positions. 
Note that for both insertions and deletions, we update the characteristic vectors in the same manner. 

\begin{figure}
    \centering

    \begin{subfigure}{\linewidth}
    \centering
    \tikzset{every picture/.style={line width=0.75pt}} 

    \begin{tikzpicture}[x=0.75pt,y=0.75pt,yscale=-1,xscale=1]
    
    \draw   (8.33,25.5) .. controls (8.33,20.31) and (12.46,16.1) .. (17.55,16.1) .. controls (22.64,16.1) and (26.77,20.31) .. (26.77,25.5) .. controls (26.77,30.69) and (22.64,34.9) .. (17.55,34.9) .. controls (12.46,34.9) and (8.33,30.69) .. (8.33,25.5) -- cycle ;
    
    \draw   (11.26,97.52) .. controls (11.26,92.33) and (15.39,88.12) .. (20.48,88.12) .. controls (25.57,88.12) and (29.69,92.33) .. (29.69,97.52) .. controls (29.69,102.71) and (25.57,106.92) .. (20.48,106.92) .. controls (15.39,106.92) and (11.26,102.71) .. (11.26,97.52) -- cycle ;
    
    \draw   (119.9,17.2) .. controls (119.9,12.01) and (124.03,7.8) .. (129.12,7.8) .. controls (134.21,7.8) and (138.33,12.01) .. (138.33,17.2) .. controls (138.33,22.39) and (134.21,26.6) .. (129.12,26.6) .. controls (124.03,26.6) and (119.9,22.39) .. (119.9,17.2) -- cycle ;
    
    \draw   (72.78,54.8) .. controls (72.78,49.6) and (76.91,45.4) .. (82,45.4) .. controls (87.09,45.4) and (91.22,49.6) .. (91.22,54.8) .. controls (91.22,59.99) and (87.09,64.2) .. (82,64.2) .. controls (76.91,64.2) and (72.78,59.99) .. (72.78,54.8) -- cycle ;
    
    \draw   (48.13,129.26) .. controls (48.13,124.06) and (52.25,119.86) .. (57.34,119.86) .. controls (62.43,119.86) and (66.56,124.06) .. (66.56,129.26) .. controls (66.56,134.45) and (62.43,138.66) .. (57.34,138.66) .. controls (52.25,138.66) and (48.13,134.45) .. (48.13,129.26) -- cycle ;
    
    \draw [color={rgb, 255:red, 208; green, 2; blue, 27 }  ,draw opacity=1 ][line width=2.25]    (26.77,25.5) -- (72.78,54.8) ;
    \draw [color={rgb, 255:red, 155; green, 155; blue, 155 }  ,draw opacity=1 ][line width=0.75]    (26.77,23.67) -- (119.9,17.2) ;
    \draw [color={rgb, 255:red, 208; green, 2; blue, 27 }  ,draw opacity=1 ][line width=2.25]    (17.55,34.9) -- (20.48,88.12) ;
    \draw [color={rgb, 255:red, 208; green, 2; blue, 27 }  ,draw opacity=1 ][line width=2.25]    (27.25,104.24) -- (49.96,123.03) ;
    \draw [color={rgb, 255:red, 208; green, 2; blue, 27 }  ,draw opacity=1 ][line width=2.25]    (57.34,119.86) -- (76.69,62.61) ;
    \draw [color={rgb, 255:red, 245; green, 166; blue, 35 }  ,draw opacity=1 ][line width=2.25]    (88.77,48.33) -- (119.9,17.2) ;
    
    \draw (12.5,17.26) node [anchor=north west][inner sep=0.75pt]   [align=left] {1};
    \draw (52.29,121.01) node [anchor=north west][inner sep=0.75pt]   [align=left] {5};
    \draw (15.43,89.28) node [anchor=north west][inner sep=0.75pt]   [align=left] {3};
    \draw (76.95,46.55) node [anchor=north west][inner sep=0.75pt]   [align=left] {2};
    \draw (124.07,8.96) node [anchor=north west][inner sep=0.75pt]   [align=left] {4};

    \end{tikzpicture}

    \caption{Graph $\graph$.}
    \end{subfigure}\vspace{2mm}

    \begin{subfigure}{\linewidth}
        \centering
        \begin{tabular}{|c|c c c c c c c c c c|}
        \hline
            ~ & (1,2) & (1,3) & (1,4) & (1,5) & (2, 3) & (2,4) & (2,5) & (3,4) & (3,5) & (4,5) \\ \hline
            1 & 1 & 1 & \textcolor{red}{1} & 0 & 0 & 0 & 0 & 0 & 0 & 0 \\ \hline
            2 & 1 & 0 & 0 & 0 & 0 & \textcolor{red}{0} & 1 & 0 & 0 & 0 \\ \hline
            3 & 0 & 1 & 0 & 0 & 0 & 0 & 0 & 0 & 1 & 0 \\ \hline
            4 & 0 & 0 & \textcolor{red}{1} & 0 & 0 & 0 & 0 & 0 & 0 & 0 \\ \hline
            5 & 0 & 0 & 0 & 0 & 0 & \textcolor{red}{0} & 1 & 0 & 1 & 0 \\ \hline
        \end{tabular}
        \caption{$\graph$'s characteristic vectors. The values changed by the edge updates are indicated in red.}
    \end{subfigure}

    \caption{An example of updating graph $\graph$ and its characteristic vectors with an orange edge insertion $(2,4)$ and a grey edge deletion $(1,4)$.}
    \label{fig:example-edge-update}
\end{figure}
For instance, imagine that we stream in one edge insertion of $(2,4)$ and one edge deletion of $(1,4)$.

For the \textbf{insertion} of the edge $(2,4)$, we add a one-hot vector $e_{(2,4)}$ to the rows $f_2$ and $f_4$. Likewise, for the \textbf{deletion} of $(1,4)$, we can add the vector $e_{(1,4)}$ to the rows $f_1$ and $f_4$.
Adding those one-hot vectors $e_{(1,4)}$ and $e_{(2,4)}$ to the appropriate characteristic vectors updates our representation of $\graph$ to reflect these updates.

The process of performing these edge updates and their impact upon the graph $\graph$ and characteristic vectors is shown in Figure~\ref{fig:example-edge-update}.



\paragraph{Computing the connected components.} Recall that \Boruvka's algorithm samples an edge from each node, merges the nodes at the endpoints of every sampled edge into supernodes, and repeats this procedure either until there only remain supernodes with no edges between them (which is, in the worst-case, after $\log_2\nodesize$ rounds have passed).

\begin{figure}
    \centering

    \begin{subfigure}{\linewidth}
    \centering
    \tikzset{every picture/.style={line width=0.75pt}} 

    \begin{tikzpicture}[x=0.75pt,y=0.75pt,yscale=-1,xscale=1]
    
    \draw   (30.49,34.96) .. controls (30.49,30.24) and (34.24,26.41) .. (38.87,26.41) .. controls (43.49,26.41) and (47.25,30.24) .. (47.25,34.96) .. controls (47.25,39.68) and (43.49,43.51) .. (38.87,43.51) .. controls (34.24,43.51) and (30.49,39.68) .. (30.49,34.96) -- cycle ;
    
    \draw   (33.15,100.44) .. controls (33.15,95.72) and (36.9,91.89) .. (41.53,91.89) .. controls (46.16,91.89) and (49.91,95.72) .. (49.91,100.44) .. controls (49.91,105.16) and (46.16,108.99) .. (41.53,108.99) .. controls (36.9,108.99) and (33.15,105.16) .. (33.15,100.44) -- cycle ;
    
    \draw   (131.93,27.41) .. controls (131.93,22.69) and (135.68,18.87) .. (140.3,18.87) .. controls (144.93,18.87) and (148.68,22.69) .. (148.68,27.41) .. controls (148.68,32.13) and (144.93,35.96) .. (140.3,35.96) .. controls (135.68,35.96) and (131.93,32.13) .. (131.93,27.41) -- cycle ;
    
    \draw   (89.09,61.6) .. controls (89.09,56.88) and (92.84,53.05) .. (97.47,53.05) .. controls (102.09,53.05) and (105.84,56.88) .. (105.84,61.6) .. controls (105.84,66.32) and (102.09,70.14) .. (97.47,70.14) .. controls (92.84,70.14) and (89.09,66.32) .. (89.09,61.6) -- cycle ;
    
    \draw   (66.67,129.3) .. controls (66.67,124.58) and (70.42,120.75) .. (75.05,120.75) .. controls (79.67,120.75) and (83.43,124.58) .. (83.43,129.3) .. controls (83.43,134.02) and (79.67,137.84) .. (75.05,137.84) .. controls (70.42,137.84) and (66.67,134.02) .. (66.67,129.3) -- cycle ;
    
    \draw [color={rgb, 255:red, 155; green, 155; blue, 155 }  ,draw opacity=1 ][line width=1.5]    (47.25,34.96) -- (89.09,61.6) ;
    \draw [color={rgb, 255:red, 208; green, 2; blue, 27 }  ,draw opacity=1 ][line width=2.25]    (38.87,43.51) -- (41.53,91.89) ;
    \draw [color={rgb, 255:red, 155; green, 155; blue, 155 }  ,draw opacity=1 ][line width=1.5]    (47.69,106.55) -- (68.33,123.64) ;
    \draw [color={rgb, 255:red, 208; green, 2; blue, 27 }  ,draw opacity=1 ][line width=2.25]    (75.05,120.75) -- (93.64,68.92) ;
    \draw [color={rgb, 255:red, 155; green, 155; blue, 155 }  ,draw opacity=1 ][line width=1.5]    (103.62,55.72) -- (131.93,27.41) ;
    \draw [color={rgb, 255:red, 74; green, 144; blue, 226 }  ,draw opacity=1 ][line width=3.75]    (46.69,36.96) -- (78.66,57.44) -- (89.31,64.26) ;
    \draw [color={rgb, 255:red, 74; green, 144; blue, 226 }  ,draw opacity=1 ][line width=3.75]    (105.62,56.94) -- (131.93,30.3) ;
    \draw [color={rgb, 255:red, 74; green, 144; blue, 226 }  ,draw opacity=1 ][line width=4.5]    (49.35,104.55) -- (70,121.64) ;
    \draw  [color={rgb, 255:red, 0; green, 0; blue, 0 }  ,draw opacity=1 ][line width=2.25]  (89.31,36.52) .. controls (113.95,35.86) and (126.93,-7.43) .. (151.24,16.54) .. controls (175.54,40.52) and (122.94,78.47) .. (98.96,81.14) .. controls (74.99,83.8) and (-6.25,47.84) .. (13.73,23.87) .. controls (33.71,-0.1) and (64.67,37.19) .. (89.31,36.52) -- cycle ;
    \draw  [color={rgb, 255:red, 0; green, 0; blue, 0 }  ,draw opacity=1 ][line width=2.25]  (23.72,83.8) .. controls (48.69,63.82) and (125.6,132.41) .. (98.3,153.39) .. controls (71,174.36) and (-1.25,103.78) .. (23.72,83.8) -- cycle ;
    
    \draw (33.82,26.77) node [anchor=north west][inner sep=0.75pt]   [align=left] {1};
    \draw (70,121.11) node [anchor=north west][inner sep=0.75pt]   [align=left] {5};
    \draw (36.48,92.25) node [anchor=north west][inner sep=0.75pt]   [align=left] {3};
    \draw (92.42,53.41) node [anchor=north west][inner sep=0.75pt]   [align=left] {2};
    \draw (135.26,19.22) node [anchor=north west][inner sep=0.75pt]   [align=left] {4};

    \end{tikzpicture}

    \caption{Graph of the supernodes we obtain after the round. Supernodes are indicated in purple with the edges we found in blue and the edges yet to be found in red.}
    \end{subfigure}\vspace{2mm}

    \begin{subfigure}{\linewidth}
            \centering
                \begin{tabular}{|c| c c c c c c c c c c|}
                \hline
                    ~ & (1,2) & (1,3) & (1,4) & (1,5) & (2, 3) & (2,4) & (2,5) & (3,4) & (3,5) & (4,5) \\ \hline
                    1+2+4 & 0 & 1 & 0 & 0 & 0 & 0 & 1 & 0 & 0 & 0 \\ \hline
                    3+5 & 0 & 1 & 0 & 0 & 0 & 0 & 1 & 0 & 0 & 0 \\ \hline
                \end{tabular}
        \caption{The supernodes' merged characteristic vectors.}
    \end{subfigure}

    \caption{An example of running a round of \Boruvka's algorithm. After sampling the blue edges we merge the characteristic vectors to achieve representations of the supernodes (in black).}
    \label{fig:example-query-round}
\end{figure}
We can easily do this with the characteristic vectors: From each characteristic vector, we choose an arbitrary nonzero entry.
In this round, it so happens that 1 samples $(1,2)$, 2 samples $(1,2)$, 3 samples $(3, 5)$, $4$ samples $(2,4)$, and 5 samples $(3,5)$. We indicate the sampled edges in blue on Figure~\ref{fig:example-query-round}.

To merge the components that were newly found to be connected by sampled edges, we sum the characteristic vectors of the nodes that were merged. In general, if we sampled an edge $(u,v)$, then the merged supernode for $u$ and $v$ will have a characteristic vector $f_{uv}$ that is equal to $f_u + f_v$ (recall that our vectors are over $\mathbb{Z}_2$, and therefore this operation generates a vector corresponding to the symmetric difference of the edges adjacent to $u$ and $v$, which is exactly the edges out of $u$ and $v$ that don't have one endpoint in $u$ and another in $v$).

In this example (shown in Figure~\ref{fig:example-query-round}), we find edges connecting $\{1,2,4\}$ and $\{3,5\}$. Thus, we set $f_{124} = f_1 + f_2 + f_4$ and $f_{35} = f_3 + f_5$.
Note that these new characteristic vectors retain the edges ($(1,3)$ and $(2,5)$) that cross between supernodes. 


\tikzset{every picture/.style={line width=0.75pt}} 

\input{figures/example_final_round}
Finally, in the next round (shown in Figure~\ref{fig:example-final-round}), both the supernode representing $\{1,2,4\}$ and the supernode 
representing $\{3,5\}$ may sample either of the edges $(1,3)$ and $(2,5)$. Depending on which
edge we sample, we can get a different spanning forest.
In either case, there is now a single supernode for our connected component, with the following 
edge incidence indicator vector resulting from summing the two remaining rows.

Thus, storing the characteristic vectors for each node, and updating them on edge insertions and deletions, allows us to compute connected components using \Boruvka's algorithm. 
Storing these characteristic vectors in a streaming context, however is infeasible. Each characteristic vector has $O(\nodesize^2) = O(\nodesize^2)$ bits, and thus the total space required to store all of them is $O(\nodesize^3)$ bits.

\paragraph{Sketching the Characteristic Vectors}
The trick Ahn \etal use to make this approach space-efficient is to maintain several $\ell_0$-sketches of each characteristic vector rather than represent the vector losslessly. Recall from Theorem~\ref{thm:l0} that each such sketch has size $O(\log^2\nodesize)$ because $\delta$ (the probability of sketch failure) is set to a constant. Since an $\ell_0$-sketch is a linear function, we can see that adding the sketches of two characteristic vectors is equivalent to first adding the vectors, and then sketching the result: $\sketch(\charvec_u + \charvec_v) = \sketch(\charvec_u) + \sketch(\charvec_v)$. A similar result holds for updating a vector after an edge insertion or deletion. So we can perform all the same operations on the sketches that we would on the characteristic vectors, and get the same result (with the exception that querying a sketch for an edge during \Boruvka's algorithm fails with constant probability).

We maintain $O(\log\nodesize)$ $\ell_0$-sketches for each characteristic vector. It is straightforward to prove that this is sufficient to find all the connected components with high probability. We describe several specific $\ell_0$-sketching algorithms in the next section.

\section{Sketching}
\label{app:sketch}
In this section we describe the design of $\ell_0$-sketches (including our new sketch \sketchname) and how these sketches can be used to compute graph connectivity.


\subsection{$\ell_0$-Sketching}
\begin{problem}[$\ell_0$-sampling] 
A vector $x$ of length $n$ is defined by an input stream of updates of the form $(i,\Delta)$ where value $\Delta$ is added to $x_i$, and the task is to sample a nonzero element of $x$ using $o(n)$ space.
\end{problem}

We denote the $\ell_0$ sketch of a vector $x$ as $\sketch(x)$.  The sketch is a linear function, \ie $\sketch(x) + \sketch(y) = \sketch(x+y)$ for any vectors $x$ and $y$. Work by Cormode et al. established that there is an $\ell_0$-sampler that succeeds with probability at least $1 - \delta$ and uses $O(\log^2(n)\log(1/\delta))$ bits of space.

Previous work by Tench et al.~\cite{graphzeppelin} introduced the \cubesketch algorithm, an improved $\ell_0$-sampler for the special case of vectors $\in \mathbb{Z}_2^n$ which we describe below:

\subsection{\cubesketch.}
\begin{figure}[!t]
\begin{small}
\begin{algorithmic}[1]
\Function {update\_sketch}{idx}
\Comment{Toggle vector index `idx'}
    \ForAll {col $\in [0, \log(1/\delta))$}
        \State col\_hash $\gets$ hash$_1$(col, idx)
        \State row $\gets 0$
        \State checksum $\gets$ hash$_2$(col, idx)
        \While {row == 0 OR col\_hash[row-1] == 0}
            \State $b_{\text{row,col}}.\alpha \gets b_{\text{row,col}}.\alpha \oplus \text{idx}$
            \State $b_{\text{row,col}}.\gamma \gets b_{\text{row,col}}.\gamma \oplus \text{checksum}$
            \State row $\gets \text{row} + 1$
        \EndWhile
    \EndFor
\EndFunction

\Function{query\_sketch}{ }
\Comment{Get a non-zero vector index}
    \ForAll {col $\in [0, \log(1/\delta))$}
        \ForAll {row $\in [0, \log (\nodesize))$}
            \If {$b_{\text{row,col}}.\gamma == hash_2$(col, $b_{\text{row,col}}.\alpha$)}
                \State \textbf{return} $b_{\text{row,col}}\alpha$
                \Comment{Found a good bucket, done}
            \EndIf
        \EndFor
    \EndFor
    \State \textbf{return} sketch\_failure
    \Comment{All buckets bad}
\EndFunction
\end{algorithmic}
\end{small}

\caption{Pseudocode for \graphzep's \cubesketch. Each edge update $(u, v)$ is converted to and from a vector index for use with the sketch.}
\label{fig:cube_code}
\end{figure}
\begin{figure}
    \def\arraystretch{1.1}
    \centering

    \begin{subfigure}[b]{\linewidth}
        \centering
        \begin{tabular}{|P{20mm}|P{20mm}|P{20mm}|}
         \hline
         $\{1, 4, 7, \textcolor{red}{9}, 11\}$ & $\{1, 4, 7, \textcolor{red}{9}, 11\}$ & $\{1, 4, 7, \textcolor{red}{9}, 11\}$ \\
         \hline
         $\{1, 4, 11\}$ & $\{1, \textcolor{red}{9}, 11\}$ & $\{1, 4, 7, \textcolor{red}{9}, 11\}$ \\
         \hline
         $\{1\}$ & $\{\textcolor{red}{9}, 11\}$ & $\{1, 7, \textcolor{red}{9}, 11\}$\\
         \hline
         $\{\}$ & $\{\}$ & $\{1, \textcolor{red}{9}\}$\\
         \hline
    \end{tabular}\vspace{1mm}

    \begin{tabular}{|l|l|l|l|l|l|l|l|l|l|l|l|l|}
        \hline
        0 & 1 & 0 & 0 & 1 & 0 & 0 & 0 & 0 & \textcolor{red}{1} & 0 & 1\\ 
        \hline
    \end{tabular}
    \caption{\cubesketch buckets (above) and the input vector (below) after adding 1 to index 9.}
    \end{subfigure}\vspace{2mm}
    
    \begin{subfigure}[b]{\linewidth}
        \centering
        \begin{tabular}{|P{20mm}|P{20mm}|P{20mm}|}
         \hline
         $\{\hcancel[red]{1}, 4, 7, 9, 11\}$ & $\{\hcancel[red]{1}, 4, 7, 9, 11\}$ & $\{\hcancel[red]{1}, 4, 7, 9, 11\}$ \\
         \hline
         $\{\hcancel[red]{1}, 4, 11\}$ & $\{\hcancel[red]{1}, 9, 11\}$ & $\{\hcancel[red]{1}, 4, 7, 9, 11\}$\\
         \hline
         $\{\hcancel[red]{1}\}$ & $\{9, 11\}$ & $\{\hcancel[red]{1}, 7, 9, 11\}$\\
         \hline
         $\{\}$ & $\{\}$ & $\{\hcancel[red]{1}, 9\}$\\
         \hline
        \end{tabular}\vspace{1mm}
        
        \begin{tabular}{|l|l|l|l|l|l|l|l|l|l|l|l|l|}
        \hline
        0 & \textcolor{red}{0} & 0 & 0 & 1 & 0 & 0 & 0 & 0 & 1 & 0 & 1\\ 
        \hline
        \end{tabular}

    \caption{\cubesketch buckets (above) and the input vector (below) after adding 1 to index 1.}
    \end{subfigure}

    \caption{Example of the structure of \cubesketch buckets. In this example, the sketch has 3 columns and 4 rows of buckets. Toggling a vector index either sets it to 1 (and thus adds it to the buckets) or sets it to 0 (thus removing it from the buckets). In this example we begin with nonzero indices $\{1, 4, 7, 11\}$. We then flip index $9$ to 1 and then flip index $1$ to 0. We can extract nonzero indices from any good buckets. That is, buckets which contain only one nonzero index.}
    \label{fig:cube-example}
\end{figure}
Given a vector $\charvec \in \mathbb{Z}_2^\veclength$, a \cubesketch consists of a matrix of $\log(\veclength)$ by $\log(1/\delta)$ \defn{buckets}.
Each bucket lossily represents the values at a random subset of positions of $\charvec$. It does so with two values \defn{$\indexsubset.\alpha$}, which is the xor of all nonzero positions, and \defn{$\indexsubset.\gamma$}, the xor of the hash of each nonzero position.  We can recover a nonzero element of $\charvec$ from bucket $\indexsubset_{i,j}$ only when a single position in $\indexsubset_{i,j}$ is nonzero. In this case we call bucket $\indexsubset_{i,j}$ \defn{good}, otherwise we say it is \defn{bad}. 
When $\indexsubset$ is good, then $\indexsubset.\alpha$ is equal to the nonzero index and $\text{hash}(\indexsubset.\alpha) = \indexsubset.\gamma$. If $\indexsubset$ is bad then with high probability $\text{hash}(\indexsubset.\alpha) \not= \indexsubset.\gamma$. 

The procedure for updating index $i$ is outlined in Figure~\ref{fig:cube_code}: for each column $c$, row $r$ is chosen with probability $1/2^r$ (using a $2$-wise independent hash function $h_j: [n] \rightarrow \log\veclength$) and $i$ is applied to each $b_{j,c}$ for $j \leq r$.
See Figure~\ref{fig:cube-example} for an example of \cubesketch's update procedure.

\subsection{\sysname's new sketch: \sketchname}
\begin{figure}[!t]
\begin{small}
\begin{algorithmic}[1]
\Function {update\_sketch}{idx}
\Comment{Toggle vector index `idx'}
    \State checksum $\gets$ hash$_2$(idx)
    \ForAll {col $\in [0, \log(1/\delta)]$}
        \State depth $\gets$ $\log_2$(hash$_1$(col, idx))

        \State $b_0.\alpha$ $\gets b_0.\alpha \oplus \text{idx}$
        \State $b_0.\gamma$ $\gets b_0.\gamma \oplus \text{checksum}$

        \State $b_{\text{depth}}.\alpha$ $\gets b_{\text{depth}}.\alpha \oplus \text{idx}$
        \State $b_{\text{depth}}.\gamma$ $\gets b_{\text{depth}}.\gamma \oplus \text{checksum}$

    \EndFor
\EndFunction

    
\end{algorithmic}
\end{small}

\caption{Pseudocode for \sysname's \sketchname update procedure. Each edge update $(u, v)$ is converted to and from a vector index for use with the sketch.}
\label{fig:cameo_code}
\end{figure}

\begin{figure}
    \def\arraystretch{1.1}
    \centering
    \begin{subfigure}[b]{\linewidth}
        \centering
        \begin{tabular}{|P{20mm}|P{20mm}|P{20mm}|}
         \hline
         $\{1, 4, 7, \textcolor{red}{9}, 11\}$ & $\{1, 4, 7, \textcolor{red}{9}, 11\}$ & $\{1, 4, 7, \textcolor{red}{9}, 11\}$ \\
         \hline
         $\{4, 11\}$ & $\{1\}$ & $\{4\}$ \\
         \hline
         $\{1\}$ & $\{\textcolor{red}{9}, 11\}$ & $\{7, 11\}$\\
         \hline
         $\{\}$ & $\{\}$ & $\{1, \textcolor{red}{9}\}$\\
         \hline
        \end{tabular}\vspace{1mm}

        \begin{tabular}{|l|l|l|l|l|l|l|l|l|l|l|l|l|}
        \hline
        0 & 1 & 0 & 0 & 1 & 0 & 0 & 0 & 0 & \textcolor{red}{1} & 0 & 1\\ 
        \hline
        \end{tabular}
    \caption{\sketchname buckets (above) and the input vector (below) after adding 1 to index 9.}
    \end{subfigure}\vspace{0.2cm}

    \begin{subfigure}[b]{\linewidth}
        \centering
        \begin{tabular}{|P{20mm}|P{20mm}|P{20mm}|}
         \hline
         $\{\hcancel[red]{1}, 4, 7, 9, 11\}$ & $\{\hcancel[red]{1}, 4, 7, 9, 11\}$ & $\{4, 7, 9, 11\}$ \\
         \hline
         $\{4, 11\}$ & $\{\hcancel[red]{1}\}$ & $\{4\}$\\
         \hline
         $\{\hcancel[red]{1}\}$ & $\{9, 11\}$ & $\{7, 11\}$\\
         \hline
         $\{\}$ & $\{\}$ & $\{\hcancel[red]{1}, 9\}$\\
         \hline
        \end{tabular}\vspace{1mm}
        
        \begin{tabular}{|l|l|l|l|l|l|l|l|l|l|l|l|l|}
        \hline
        0 & \textcolor{red}{0} & 0 & 0 & 1 & 0 & 0 & 0 & 0 & 1 & 0 & 1\\ 
        \hline
        \end{tabular}

    \caption{\sketchname buckets (above) and the input vector (below) after adding 1 to index 1.}
    \end{subfigure}
    
    \caption{Example of the structure of \sketchname buckets. The setup for this example is the same as for \cubesketch (see Figure~\ref{fig:cube-example}). However, in \sketchname we update fewer buckets for each index. As we can see in this figure only 2 buckets per column contain a given index. The result is an asymptotic improvement to update time and an increase in the number of good buckets.}
    \label{fig:cameo-example}
\end{figure}
Essential to the performance \sysname achieves is our new $\ell_0$-sampler called \sketchname. \sketchname improves over \cubesketch with a new update procedure that is a factor $\log n$ faster to update and reduces space usage by a constant factor via a refined analysis. All other details, including the query procedure, remain unchanged.

\paragraph{Update procedure.} \sketchname maintains the same matrix of buckets as \cubesketch, but uses a simpler and faster update procedure as shown in Figure~\ref{fig:cameo_code}. When performing an update $u$, for each column $c$ we choose a row $r$ independently with probability $1/2^r$. Unlike \cubesketch which updates all rows $[0, r]$ in the column, we apply $u$ to only $b_{0,c}$ and $b_{r,c}$. 

It is straightforward to see that \sketchname maintains the correctness and probabilistic bounds of \cubesketch. Assuming they use the same hash functions, if a bucket is good in the latter it must also be good in the former.


\restatecameoupdatethm
\begin{proof}
    The proof follows from the analysis of \cubesketch. If a \cubesketch and \sketchname share the same randomness, then the \cubesketch returning a valid edge implies that the \sketchname must also return a valid edge. This is because each \sketchname bucket contains a subset of the contents of the same \cubesketch bucket. If the \cubesketch column is good then there exists a deepest (largest row value) good bucket $b_{i,j}$. \sketchname's $b_{i,j}$ must be identical to \cubesketch's. This is because each nonzero index in \sketchname appears at its deepest row and because of the subset property.

    The only exception to this property is if \sketchname returns an incorrect edge due to a checksum error. However, this happens with polynomially small probability and thus does not violate the proof.
\end{proof}

Theorem~\ref{thm:cameo_update} demonstrates that \sketchname reduces the CPU burden of performing updates and supports Claim~\ref{claim:all}.\ref{itm:cpu}.

\paragraph{Reduced constant factors.} In \graphzep, Tench \etal use $72 \log (1/\delta)\log n$ bytes of space to guarantee a failure probability of at most $\delta$ when sketching a vector of length $n < 2^{64}$ using \cubesketch. Via a careful constant-factor analysis, we can show that \sketchname can match this failure probability with significantly less space: 

\restatecameospacethm

The intuition behind this result is that for every sketch column there exists a depth $d$ such that fewer than 7 of the stored elements exist below it. Given this knowledge, we can use fully independent analysis upon this subset of the buckets to achieve a tighter lower bound on the probability of success.
We omit the complete proof due to space constraints. 

Theorem~\ref{thm:cameo_cols} immediately implies a space savings of up to 90\% compared to \cubesketch and thus supports Claim~\ref{claim:all}.\ref{itm:space}. In our implementation, we conservatively choose to use slightly more space than this theorem requires to reduce the failure probability further. Still, our implementation requires only $2/7$ths of the space used in \graphzep~\cite{graphzeppelin} (see Section~\ref{sec:system} for details).

\section{Design of \Treename}
\label{app:pht}
\begin{figure}
    \centering
    \includegraphics[width=0.45\textwidth]{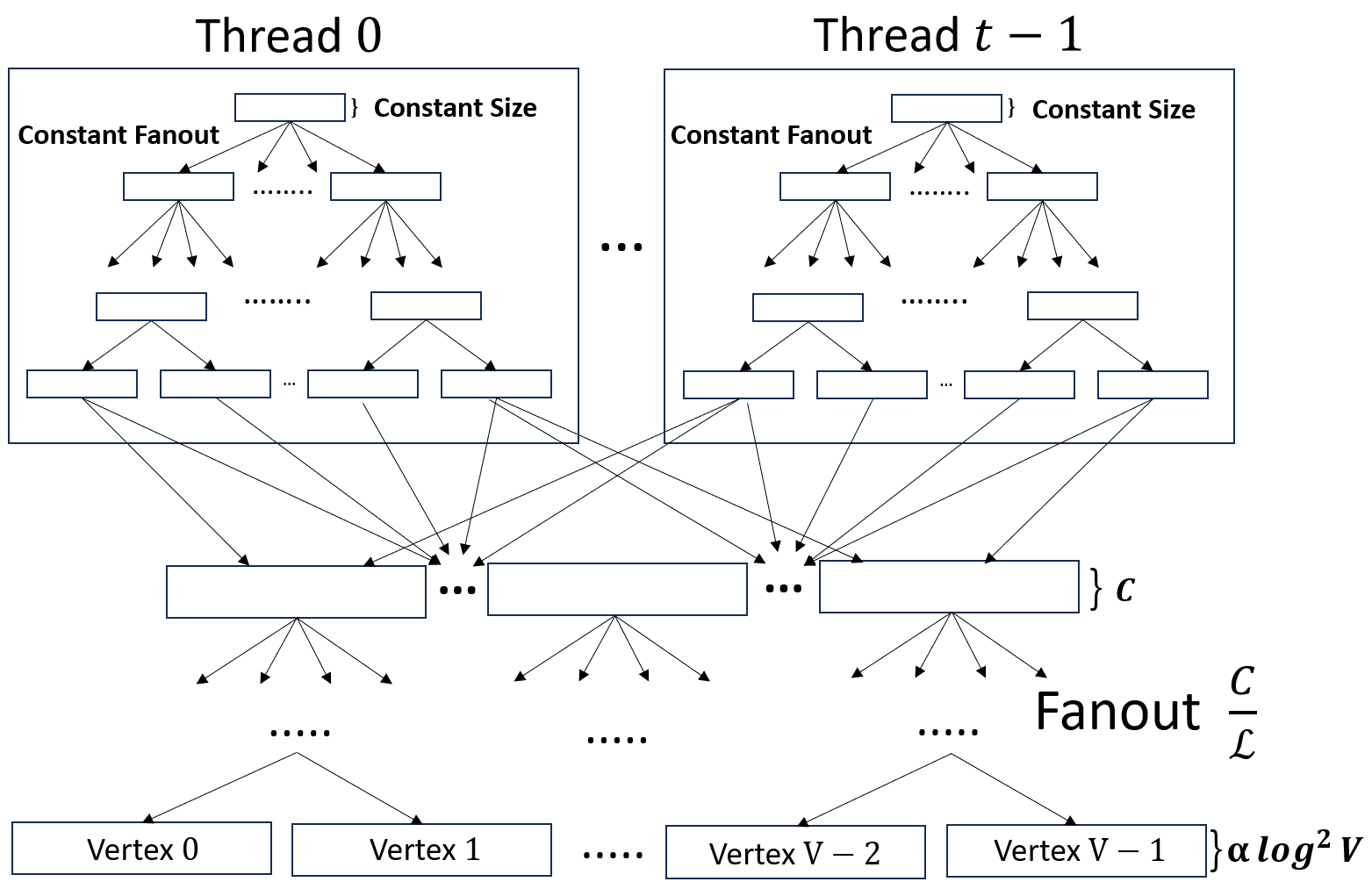}
    \caption{The Pipeline Hypertree data-structure.}
    \label{fig:pht}
\end{figure}

To make vertex-based batching fast, we design the \defn{\treename}, which is a simplified and parallel variant of the buffer tree~\cite{Arge03} designed to minimize cache line misses and thread contention. The \treename receives arbitrarily ordered stream updates and consolidates them into vertex-based batches.

The \treename is logically structured as a directed acyclic graph (DAG) $D$ with $O(\nodesize + t)$ nodes and $O(\nodesize +t)$ edges, managed by $t$ threads. The nodes are partitioned into $O(\log_{\cachesize/\cachelinesize}\nodesize)$ \defn{levels}, where $\cachesize$ denotes the size of L3 cache and $\cachelinesize$ denotes the size of an L3 cache line. Nodes of $D$ in levels 0 through $\rho$ (for some constant $\rho$) are \defn{thread local}: for each node in level $l \leq \rho$ there is exactly one thread that can read from or write to the node.

Specifically, there are $t$ nodes in level 0 and each thread owns a unique level 0 node. Nodes in levels 0 to $\rho-1$ each have constant y children, and each thread has exclusive access to all descendants of a level 0 node (up to level $\rho$). For levels $\geq \rho$, each node has $O(\cachesize/\cachelinesize)$ children. These \defn{global nodes} in levels $>\rho$ can be read from or written to by any thread. At the lowest level of $D$, there are $\nodesize$ leaf nodes, one for each vertex in the input graph $\graph$. Note that there is exactly one path from any internal node in the DAG to any leaf.
Each node in $D$ has a buffer. Thread-local nodes have a buffer of constant capacity, leaf nodes have buffers of capacity $O(\alpha\log^2\nodesize)$ updates and other global nodes have buffers of capacity $\cachesize$ updates. It is important that the leaf nodes have a limited capacity to ensure that the size of the \treename does not grow beyond that of the \sketchnames. By this construction, each leaf has size $O(\alpha \log^3 \nodesize)$ bits (a constant factor $\alpha$ greater than that of a \sketchname). Figure~\ref{fig:pht} illustrates the design.

The API for the \treename supports two operations: insert(update) and force\_flush(). During an insert, update $e = (u,v)$ a thread inserts $(u,v)$ and $(v,u)$ into the buffer of its root node. When any non-leaf buffer fills, each update $(j,k)$ in the buffer is moved to the child on the path to leaf $j$. We call this operation a \defn{flush}. Moving elements to children may fill the buffer of a child and cause the child to flush as well. Finally, when a leaf buffer fills, its contents are packaged as a batch and sent over the network for update processing. The force\_flush operation allows a user to force all buffered updates to be moved through the \treename and into the \sketchnames. This is accomplished by performing a flush operation on each node of $D$ in a breadth first search pattern.

The \treename design avoids unnecessary cache misses and thread contention. We only accesses a node of the DAG $D$ when we have many updates to move into the node. As a consequence, the amortized cost of placing a single update into the vertex-based batch is less than a single cache miss. This is why we do not simply use a hash table (which costs at least one cache miss per update) to perform the batching. Additionally, the top levels of $D$ are kept thread-local to avoid synchronization. Keeping shared data lower in $D$ reduces the likelihood that two threads touch the same node simultaneously.

It is immediate from the construction of the \treename that each update is moved $O(\log_{\cachesize/\cachelinesize}\nodesize)$ times and that the total size of the data structure is $O(\nodesize\log^3\nodesize)$ bits. Thus, \treename achieves Claim~\ref{claim:all}.\ref{itm:ram} and does not violate Claim~\ref{claim:all}.\ref{itm:space}.
In Appendix ~\ref{subsec:ingestor_design} we describe our implementation of the \treename, and in Section ~\ref{subsec:scaling} we show it is crucial for fast stream ingestion.




\section{Proofs of Theorems ~\ref{thm:queries} and ~\ref{thm:comp}}
\label{app:proofs}
\commthm*
\begin{proof}
In the worst case, the input stream sends $O(\alpha\gamma\nodesize\sketchsize)$ edge updates such that each vertex $u \in \nodes$ is an endpoint of exactly $\alpha\gamma\sketchsize$ of these edges. Then it issues a connectivity query. At this point, each leaf (which has size $\alpha\sketchsize$ is exactly a $\gamma$-fraction full, and so sends the contents of each leaf as a vertex-based batch to a distributed worker. The worker responds with a sketch delta of size $\sketchsize$, so the average communication induced per edge update is $O(\alpha\gamma)$.

Note that no query can induce more communication that this - if any leaf buffer is less than a $\gamma$-fraction full, the updates it contains are processed locally and incur no communication. If any leaf buffer is more than a $\gamma$-fraction full, the average communication per update in that buffer is lower (since the sketch delta is always a fixed size).

If the input stream alternates sending $O(\alpha\gamma\sketchsize\nodesize)$ edge updates as above and issuing connectivity queries, the average communication cost per stream update is $O(\alpha\gamma)$, regardless of stream length.
\end{proof}

\compthm*
\begin{proof}
The space bound is immediate from the size of $\sketch(\graph)$ and the \treename. 

Assuming $N_i = \omega(\alpha\nodesize \log^3\nodesize)$, the average number of updates per graph vertex is $\omega(\alpha\log^3\nodesize)$. If a constant fraction of all graph vertices receive fewer than $O(\alpha\gamma\log^2\nodesize)$ updates (and therefore process these updates locally) the total CPU work done on the main node to process all updates is $O(\log\nodesize \cdot \nodesize \alpha\gamma\log^2\nodesize + \log_{\cachesize/\cachelinesize}\nodesize \cdot \nodesize \cdot \alpha\log^3\nodesize)$ and therefore the average work per update is $\log_{\cachesize/\cachelinesize}$. On the other hand, if $\streamlength_i = O(\alpha\nodesize\log^3\nodesize)$ in the extreme case all batches may be processed locally and then the CPU work per update is $O(\log\nodesize)$.
\end{proof}

\section{\sysname Implementation}
\label{app:system}

\sysname must handle two tasks: stream ingestion, where edge updates from the input stream are compressed into the graph sketch; and query processing, where connectivity queries are computed from the graph sketch (or sometimes from auxiliary query-accelerating data structures, described below). In this section we describe the implementation of \textsc{StreamIngestor}, which handles stream ingestion, and \textsc{QueryProcessor}, which handles answering queries.

Processing stream updates into the graph sketch is a computationally intensive process: for a moderately sized data-set with $2^{18}$ vertices, applying a single edge update requires evaluating 184 hash functions. \sysname farms out this computationally intensive portion of the workload to worker nodes while the other portions of stream ingestion, including update buffering and sketch storage, remain the responsibility of the main node. \sysname uses $164\nodesize*(\log^2\nodesize-\log\nodesize)$B of space on the main node to store the sketches and the \treename. On the worker nodes, \sysname requires storage for a single sketch and batch per CPU. Thus, a worker node with $t$ threads requires $t \cdot 164(\log^2\nodesize-\log\nodesize)$ bytes. On a billion-vertex graph, each worker thread requires only 64 KiB of RAM. When computing $k$-connectivity, all of the above costs are multiplied by a factor $k$.


\subsection{Distributed Communication}
We use OpenMPI for message passing over the network.
The main node stores the graph sketch and the \treename. As edge updates arrive from the input stream, the main node inserts them into the \treename and performs flushes when necessary. When a leaf fills, the main node sends the contents of the leaf (a vertex-based batch) as an OpenMPI message to a worker node. 


\subsection{\textsc{StreamIngestor} Design}
\label{subsec:ingestor_design}


\sysname's \textsc{StreamIngestor} processes information from the input stream into the graph sketch.  Stream updates go through three key steps on their way to being applied to the graph sketch: first the updates are collected into vertex-based batches in the \treename on the main node, then these batches are processed into sketch deltas by the worker nodes, and finally each sketch delta is added to the graph sketch on the main node. 
A high level description of \textsc{StreamIngestor} is outlined in Figure~\ref{fig:sys_diagram}. We summarize the design of each component below.

\textbf{\Treename.}

We chose the parameters of the \treename to balance insertion performance and memory footprint. Levels 0, 1, and 2 in the \treename are thread-local. Each thread has a private copy of all thread-local nodes it owns. Each thread owns eight level-0 nodes. 
At level i for $i \in [0,2]$ each node has a buffer of size $8\cdot2^i$KiB and a fan-out of $16\cdot 2^i$. Levels 3 and 4 are global levels: any thread may read to or write from any node in these levels. Each level 3 node has a buffer of size $\frac{512 \cdot \nodesize}{8\cdot2^3}$KiB and a fan-out of $\frac{\nodesize}{8\cdot2^3}$. There are exactly $\nodesize$ nodes in level 4 node and each has a buffer of size $\alpha\log^2\nodesize$ (or $k\alpha\log^2\nodesize$ for k-connectivity).

Updates are assigned to threads arbitrarily. Our implementation reads streams from a file in parallel; whichever thread reads the update is responsible for buffering it. 


When a leaf buffer becomes full, its contents are placed into the Work Queue for processing by worker nodes.

\textbf{Work Queue.}
\sysname's main node uses two types of threads to send vertex-based batches to worker nodes: Graph Insertion threads produce batches of updates from the \treename, and Work Distributor threads send these batches over the network for distributed processing. Data flow between Graph Insertion and Work Distributor threads is synchronized by the \defn{Work Queue}, a many-producer, many-consumer queue for seamless communication between Graph Insertion threads (producers) and the Work Distributor threads (consumers).

The Work Queue uses two linked lists and two mutexes/condition variables to coordinate simultaneous operations. Threads interacting with the Work Queue only need to hold locks for a constant amount of time, allowing us to achieve a high degree of parallelism. This is because both Work Queue operations perform only pointer swapping within their critical regions to move vertex-based batches into and out of the queue.

\textbf{Generating sketch deltas.}
\label{subsec:cubesketch}
To generate a sketch delta from a batch of updates for vertex $u$, the worker node applies each update to an initially empty vertex sketch $S_u(\varnothing)$. 
We summarize the process of updating \sketchnames (see Fig.~\ref{fig:cameo_code}).
A vertex sketch consists of $\log_{3/2}(\nodesize)$ \sketchnames that are modified by each edge update. Each \sketchname stores a $O(\log(\nodesize^2)) \times 2$ matrix of ``buckets''. Finally, each bucket is defined by two variables $\alpha$ and $\gamma$.



Performing a \sketchname update requires performing $1 + (\log(1 / \delta) = 2) = 3$ hash calls (we use xxHash~\cite{xxhash}), one for each column and one to determine the checksum. Once we have computed the hash values, we complete the update with four bitwise XORs. The cost of updating a sketch is dominated by the hash calls. To update a sketch with one edge update requires $3 \times \log_{3/2} \nodesize$ hash calls, for a graph on $2^{18}$ vertices this is 92, and we must update two sketches per edge insertion for a total of 184 hash calls.


\subsubsection{Extending to k-connectivity}
\sysname can apply updates to k-connectivity sketches using essentially the same method as described for connectivity above. The size of a vertex-based batch (and consequently the size of leaf node buffers in the \treename) is increased by a factor $k$ to $164 \cdot \alpha k \nodesize(\log^2(\nodesize)-\log(\nodesize))$ bytes. To process a batch, the distributed worker produces sketch deltas for all $k$ connectivity sketches, concatenates them, and sends them back to the main node.

\subsection{\textsc{QueryProcessor} Design}

\sysname users may issue global connectivity or batched reachability queries while the stream is being processed. Issuing a query provides an answer accurate to the graph constructed by the stream up to the point the query was issued. To provide this accuracy requires pausing stream ingestion while processing the query so that all the sketches are in the correct state.

When a query is issued, if \sysname has a valid \dsuname instance then it uses it to answer the query as described in Section~\ref{subsec:dsu}. We show in Section~\ref{sec:experiments} that answering queries this way is extremely fast. If there is no valid \dsuname when the query is issued, \sysname instead flushes all pending updates out of the \treename, applies them to its graph sketch, and then runs \Boruvka's algorithm to compute connectivity.

\textbf{Flushing and applying updates.}
All pending updates must be applied to the graph sketch before the query can be computed. When a query is issued, if there is no valid \dsuname, Graph Insertion threads immediately flush the \treename, forcing all updates into its leaves. As described in Section~\ref{subsec:queries}, \sysname distributes the work of processing the updates in each leaf provided the leaf is full enough; nearly-empty leaves are instead processed locally on the main node.

We chose a $4\%$ fullness threshold for this policy because it gave the best performance on our cluster; \sysname users can change this default threshold value if desired. By Theorem~\ref{thm:queries} this can in the worst case lead to a $25\times$ increase in network communication. We demonstrate in Section~\ref{sec:experiments} that the network communication blowup on real input streams is within this bound (and is typically much lower).

\textbf{\Boruvka's algorithm with sketches.}
Once all sketch deltas have been merged into the graph sketch, we run \Boruvka's algorithm 
to find the spanning forest of the graph defined by the input stream.
\Boruvka's algorithm proceeds over the \sketchnames as described in Section~\ref{sec:sketch}.

\subsection{\dsuname: Reusing Prior Queries.}
\label{subsec:dsu}
\sysname is capable of reusing results from prior queries to drastically reduce latency for future queries. In response to a connectivity query $q$, \sysname uses \Boruvka's algorithm to produce a spanning forest of the graph. \sysname retains the information from this spanning forest in a data structure we call \dsuname and uses it for future queries. 


\paragraph{Data structure.} \dsuname consists of a union-find data structure encoding the connected components of the spanning forest, and a hash table containing the edges of the spanning forest. Both of these data structures are compact: they require $O(\nodesize)$ space. After the query $q$, when edge update $e = (u,v)$ arrives from the input stream, in addition to inserting $e$ into the \treename, \sysname also uses $e$ to update \dsuname: if $e$ is an edge insertion and $u$ and $v$ are not in the same connected component, a Graph Insertion thread merges their components in \dsuname's union find in $O(\mathcal{A}(\nodesize))$ time, where $\mathcal{A}$ is the inverse Ackerman function~\cite{ackermann}. It also adds $e$ to the hash table in $O(1)$ expected time. 

\paragraph{Query acceleration.} If a second query $q'$ is issued later, \sysname can use \dsuname rather than recompute graph connectivity from scratch, which would take $O(\nodesize \log^3(\nodesize))$ time. For global connectivity queries, \sysname returns \dsuname's spanning forest in $O(\nodesize)$ time and for batched reachability queries it uses \dsuname's union-find data structure to compute reachability for all of the $m$ vertex pairs in $O(m\mathcal{A}(\nodesize))$ time. We demonstrate experimentally in Section~\ref{subsec:queryex} that answering queries using \dsuname is several orders of magnitude faster than the sketch \Boruvka algorithm.

\paragraph{\dsuname validity.} However, if one of the edges $e = (u,v)$ in \dsuname's spanning forest is deleted after the query, \dsuname does not retain enough information to determine whether or not vertices $u$ and $v$ are still connected and, if they are, to find a replacement edge. Recovering this information is only possible by running the sketch \Boruvka algorithm. When a spanning forest edge is deleted in this way, we say that \dsuname has become \defn{invalid} meaning it is no longer useful for answering connectivity queries. While \sysname maintains a valid \dsuname, as each edge update $e = (u,v)$ arrives if $e$ is a deletion \sysname checks whether or not $e$ is in \dsuname's hash table. If it is, \sysname discards \dsuname as invalid. In the worst case, an adversarial input stream would render \dsuname invalid immediately after each query by deleting a key edge but this behavior is pathological and, as we show in Section~\ref{subsec:scaling}, unlikely to occur in real data streams. For practical data streams it is more likely that \dsuname remains valid for a while after a query, during which time subsequent queries can be answered very quickly.

\paragraph{\dsuname implementation.}
\label{app:greedycc_impl}
\dsuname consists of a union-find data structure encoding the connected components of the spanning forest, and a hash table containing the edges of the spanning forest. Both of these data structures are compact: they require $O(\nodesize)$ space. After the query $q$, when edge update $e = (u,v)$ arrives from the input stream, in addition to inserting $e$ into the \treename, \sysname also uses $e$ to update \dsuname: if $e$ is an edge insertion and $u$ and $v$ are not in the same connected component, a Graph Insertion thread merges their components in \dsuname's union find in $O(\mathcal{A}(\nodesize))$ time, where $\mathcal{A}$ is the inverse Ackerman function~\cite{ackermann}. It also adds $e$ to the hash table in $O(1)$ expected time. 

If a second query $q'$ is issued later, \sysname can use \dsuname rather than recompute graph connectivity from scratch, which would take $O(\nodesize \log^3(\nodesize))$ time. For global connectivity queries, \sysname returns \dsuname's hash table in $O(\nodesize)$ time and for batched reachability queries it uses \dsuname's union-find data structure to compute reachability for all of the $k$ vertex pairs in $O(k\mathcal{A}(\nodesize))$ time. We demonstrate experimentally in Section~\ref{subsec:queryex} that answering queries using \dsuname is several orders of magnitude faster than the sketch \Boruvka's algorithm.

However, if one of the edges $e = (u,v)$ in \dsuname's spanning forest is deleted after the query, \dsuname does not retain enough information to determine whether or not vertices $u$ and $v$ are still connected and, if they are, to find a replacement edge. Recovering this information is only possible by running the sketch \Boruvka's algorithm. When a spanning forest edge is deleted in this way, we say that \dsuname has become \defn{invalid} meaning it is no longer useful for answering connectivity queries. While \sysname maintains a valid \dsuname, as each edge update $e = (u,v)$ arrives if $e$ is a deletion \sysname checks whether or not $e$ is in \dsuname's hash table. If it is, \sysname discards \dsuname as invalid. In the worst case, an adversarial input stream would render \dsuname invalid immediately after each query by deleting a key edge but this behavior is pathological and, as we show in Section~\ref{subsec:scaling}, unlikely to occur in real data streams. For practical data streams it is more likely that \dsuname remains valid for a while after a query, during which time subsequent queries can be answered very quickly.
In Section~\ref{subsec:queryex} we show that \dsuname reduces query latency by up to four orders of magnitude.

\section{More Experiments}
In this section we present several additional experiments, as well as additional detail about several experiments described in the main body of the paper.
\subsection{Dense Graph Survey: Further Analysis}
\label{app:dense}
In the introduction, we argue via Figure~\ref{fig:dense}  that the lack of large, dense datasets in academic papers is likely a selection effect. The figure illustrates that graph datasets from a variety of applications and sourced from several collections rarely are larger than roughly 16GB. 

\begin{figure}
\centering
\begin{subfigure}{0.5\textwidth}
    \includegraphics[width=\textwidth]{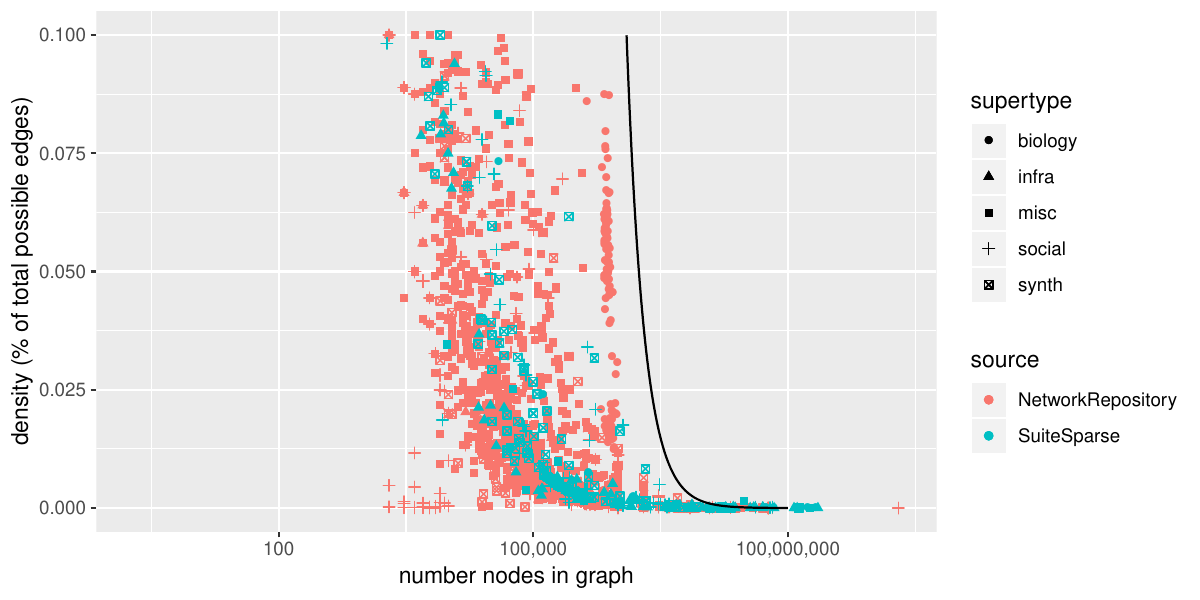}
    \caption{Graphs of up to $0.1\%$ maximum density.}
    \label{fig:densezoom}
\end{subfigure}

\begin{subfigure}{0.5\textwidth}
    \includegraphics[width=\textwidth]{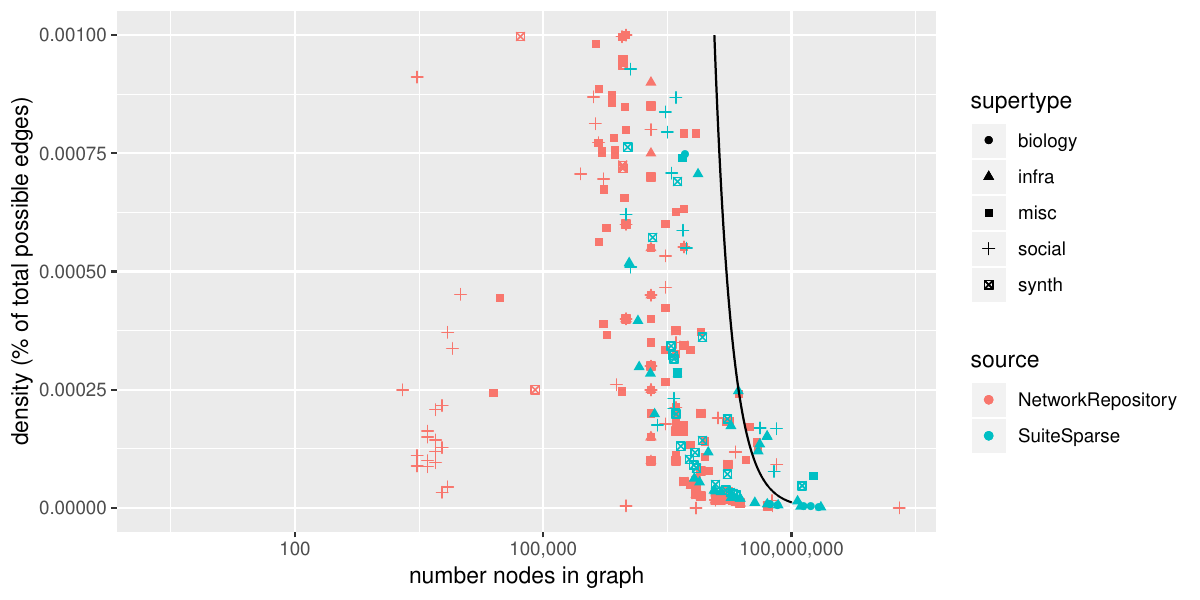}
    \caption{Graphs of up to $0.001\%$ maximum density.}
    \label{fig:densezoomer}
\end{subfigure}
\caption{Alternative visualizations of the graph datasets from NetworkRepository and SuiteSparse, emphasizing the relationship between vertex count, maximum observed density and relatively small representation size (16GB).}
\end{figure}

Figures ~\ref{fig:densezoom} and ~\ref{fig:densezoomer} are alternative visualizations of the same collections of graph datasets. In Figure ~\ref{fig:densezoom} we restrict the y axis to show graphs whose density is at most $0.1\%$ of all possible edges, and in Figure ~\ref{fig:densezoomer} graphs of at most $0.001\%$ density. The selection effect is even easier to see in the plots: as vertex count increases, density decreases such that the adjacency list size is not more than low double-digit gigabytes. 

\subsection{Correctness}
The connectivity sketch underlying \sysname is a randomized algorithm which succeeds with high probability - that is, the probability of it returning an incorrect spanning forest is bounded above by $1/\nodesize^c$ for some constant $c > 1$, which is a function of the sketch failure parameter $\delta$. We set $\delta = 0.01$ and it is straightforward to show that consequently $c>2$. As a result, the analytical failure probability of the connectivity algorithm for $\nodesize > 10,000$ is less than $10^{-8}$ regardless of the input stream.

To test correctness empirically, we compared the spanning forests output by \sysname with those computed from an adjacency matrix. We repeated this experiment 1000 times each for the kron17, p2p-gnutella, rec-amazon, google-plus, and web-uk streams. No failures were ever observed.

\subsection{\sysname k-connectivity performance}
\label{app:kconnect_full}
Table ~\ref{tab:kconnect_full} summarizes \sysname's performance for computing $k$-connectivity for a variety of datasets and values of $k$. All experiments were run with 32 distributed workers for a total of 512 threads.
\begin{table*}[t]
\begin{center}
\def\arraystretch{1.1}
\footnotesize
\caption{\sysname performance for computing k-connectivity on a variety of real-world and synthetic datasets. N/A entries indicate the experiment was not run because the sketch was larger than available RAM.}
\label{tab:kconnect_full}
\begin{tabular}{ |c|cccc|cccc| } 
 \hline
	&  \multicolumn{4}{|c|}{Insertions per second (millions)}	&	\multicolumn{4}{|c|}{Memory Consumption (GiB)} \\
\hline

Dataset & 	k = 1	& 	k = 2	& 	k = 4	& 	k = 8	& 	k = 1	& 	k = 2	& 	k = 4	& 	k = 8 \\
 \hline
ca-citeseer		& 17.29	& 	8.99		& 4.71	& 	2.37		& 14.51	& 	24.99		& 42.87	& 	79.57 \\
 \hline
{google-plus}	& 	108.6		& 65.4	& 	37.8		& 2.06	& 	9.68		& 13.62	& 	21.99		& 34.80 \\
 \hline
{p2p-gnutella}	& 	14.79	& 	8.13	& 	4.79	& 	2.48	& 	7.22	& 	9.16	& 	13.26	& 	19.44 \\
 \hline
{rec-amazon}	& 	13.55	& 	7.64	& 	4.48	& 	2.28	& 	8.68	& 	12.59	& 	19.51	& 	30.11 \\
 \hline
{web-uk}	& 	91.5	& 	54.7	& 	2.87	& 	14.95	& 	11.17	& 	16.38	& 	28.62	& 	48.58 \\
 \hline
{kron13}	& 	177.7	& 	113.3	& 	46.9	& 	23.67	& 	5.75	& 	5.94	& 	6.68	& 	7.49 \\
 \hline
{kron15}	& 	314.3	& 	197.7	& 	102.6	& 	50.86	& 	7.72	& 	10.02	& 	15.46	& 	20.73 \\
 \hline
{kron16}	& 	329.7	& 	207.3	& 	105.2	& 	52.67	& 	10.08	& 	14.43	& 	25.16	& 	41.32 \\
 \hline
{kron17}	& 	338.5	& 	200		& 101.5	& 	50.76	& 	15.25	& 	24.40	& 	46.49	& 	83.39 \\
 \hline
{erdos18}	& 	309.7	& 	189	94.7	& 	94.7 & N/A &	26.04	&  	45.90	& 	88.90	& N/A \\	
 \hline
{erdos19}	& 	234.1	& 	170.9	& N/A & N/A & 			50.98	& 	93.20	& N/A & N/A \\		
 \hline
{erdos20}	& 	245.8	& N/A & N/A & N/A &				105.93		& N/A & N/A & N/A \\
 \hline
 \hline
 &	\multicolumn{4}{|c|}{Query Latency (seconds)}		&	\multicolumn{4}{|c|}{Network Communication (GiB)} \\
\hline
Dataset  & 	k = 1	& 	k = 2	& 	k = 4	& 	k = 8	& 	k = 1	& 	k = 2	& 	k = 4	& 	k = 8 \\
  \hline
ca-citeseer	& 	2.447		& 7.55		& 20.25	& 38.83	& 0.11	& 	0.207	& 	0.408		& 0.808 \\
 \hline
{google-plus} & 1.03		& 3.664		& 11.54	& 38.51	& 	4.91	& 	7.798	& 	13.34	& 	24.74 \\
 \hline
{p2p-gnutella} & 	0.42	& 	1.301	& 	3.29	& 	8.548 & 	0	& 	0	& 	0	& 	0 \\
 \hline
{rec-amazon} & 	0.87	& 	2.323	& 	4.296	& 	8.22 & 	0	& 	0	& 	0	& 	0 \\
 \hline
{web-uk} & 	1.13	& 	3.715	& 	5.682	& 	10.91	& 	4.8	& 	8.322	& 	15.41 & 	29.60 \\
 \hline
{kron13}	& 	0.06	& 	0.219	& 	0.7647	& 	2.866	& 	1.72	& 	1.93	& 	2.061	& 	2.287 \\
 \hline
{kron15} & 	0.284	& 	1.076	& 	3.372	& 	13.95	& 	26.8	& 	29.5	& 	30.22	& 	31.68 \\
 \hline
{kron16} & 	0.581	& 	1.95208	& 	7.02079	& 	25.6338	& 	106.6	& 	116.864	& 	118.562	& 	121.887 \\
 \hline
{kron17} & 	1.27	& 	5.02412	& 	16.1608	& 	65.5716	& 	425.2	& 	464.981	& 	468.886	& 	476.598 \\
 \hline
{erdos18}  &		3.15	 & 	8.71687	& 	41.5415	& N/A & 		545	& 	597.472	& 	605.368 & N/A \\	
 \hline
{erdos19} &			6.97	& 	18.9148	& N/A & N/A &			551.4	& 	607.217 & N/A & N/A  \\		
 \hline
{erdos20} &			14.1	& N/A & N/A & N/A &				1377	& N/A & N/A & N/A 			\\
 \hline
 
\end{tabular}
\end{center}
\end{table*}

\subsection{Comparison to \graphzep}
\label{subsec:gzcomp}

\begin{figure}
    \centering
    \includegraphics[width=.45\textwidth]{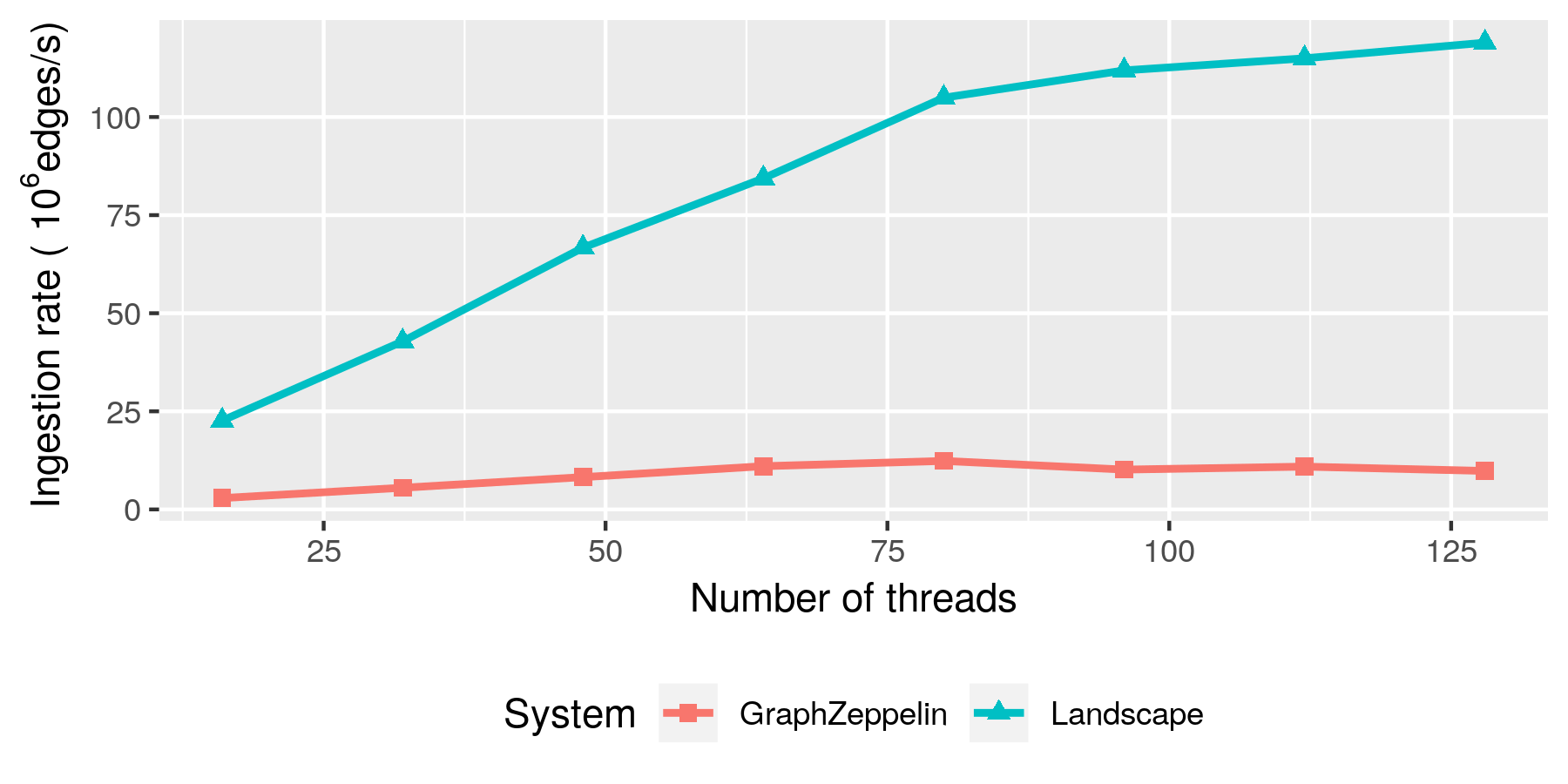}
    \caption{Single-machine scaling performance of \graphzep and \sysname. \graphzep's maximum ingestion rate is less than 15 million. \sysname avoids this bottleneck due to \sketchname and the \treename.}
    \label{fig:scaling_gz}
\end{figure}

For completeness, we compare \sysname's ingestion rate and query latency against GraphZeppelin because it uses sketching and is opimized for dense graphs. The following experiments were run on a single AWS c6a.48xlarge instance (with {192} cores) for direct comparison since GraphZeppelin is a single-machine system.

We evaluated both systems' ingestion rates on the kron17 stream, varying the number of threads. Figure \ref{fig:scaling_gz} summarizes the results. We see that even at a small number of threads \sysname is almost an order of magnitude faster than \graphzep; this is a consequence of the reduced update time of \sketchname. Further, \graphzep fails to scale beyond 80 threads, with a maximum ingestion rate of less that 15 million updates/sec. To explain this, we also compared the throughput of each systems' buffering systems in isolation {on our cluster's main node}. Since \graphzep's in-RAM buffering data structure is not designed to minimize L2 cache misses, its performance is nearly two orders of magnitude slower than sequential RAM bandwidth (and therefore much slower than \sysname).
In contrast, the throughput of \sysname's \treename increases to over 500 million updates/sec. 

We also repeated our query latency experiment on \graphzep, which took roughly 1 second to compute each global connectivity or batched reachability query. While this is slightly lower latency on the first query in a burst than \sysname, \sysname's \dsuname heuristic results in four orders of magnitude lower latency on subsequent queries of both types.

We conclude that \sysname achieves greater performance than \graphzep, because of \sketchname, the \treename, and \dsuname.

\section{Related Work.}
\label{sec:related}
\paragraph{Single-Machine Streaming Graph Systems.}
Many existing single-machine graph stream processing systems are optimized for the \defn{batch-parallel} model where stream updates arrive in large batches consisting entirely of insertions or entirely of deletions. 
Some batch-parallel systems answer queries synchronously, meaning they stop ingesting stream updates while computing queries~\cite{terrace,busato2018hornet,ediger2012stinger,murray2016incremental,sengupta2016graphin,sengupta2017evograph}. Others periodically take ``snapshots'' of the graph during ingestion that enable asynchronous query evaluation~\cite{aspen,cheng2012kineograph,iyer2015celliq,iyer2016time,macko2015llama}.
Some systems are designed to process only streams of edge insertions~\cite{jetstream,graphbolt}.

\paragraph{Distributed Streaming Graph Systems.}
Kineograph~\cite{kineograph} is a general-purpose framework for incremental graph algorithms which uses snapshots for efficient incremental computation. GraphTau~\cite{graphtau} uses a graph snapshot method implemented on top of Apache Spark's RDDs for graph analytics in the sliding window model. Kickstarter~\cite{kickstarter} is a runtime technique for incremental monotonic graph algorithms which trims approximate vertex values when edges are deleted for efficient result updating. It is implemented in the distributed graph processing system ASPIRE\cite{aspire}. While Kickstarter is a potentially interesting point of comparison, its source code is not publically available.
CellIQ~\cite{iyer2015celliq} is a distributed graph system optimized for solving the connected components problem on cellular network graphs in the sliding window model. 
Tegra~\cite{iyer2021tegra} is designed to compute a variety of time-window graph analytics on dynamic graphs. 
DiStinger~\cite{distinger} is a distributed extension of Stinger~\cite{ediger2012stinger} which is optimized for solving PageRank on sparse graph streams.

\section{Analysis of \sketchname's Space Usage}
\label{app:cameo-proofs}
This section proves Theorem~\ref{thm:cameo_cols}. Our analysis discusses a single sketch column and then generalizes to a sketch with multiple columns. Since all buckets have the same column we drop the column from our notation and refer to the column $0$, row $r$ bucket as $b_r$.

\begin{figure}[t]
    \centering
    \begin{subfigure}{0.4\linewidth}
        \centering
        \includegraphics[scale=0.5]{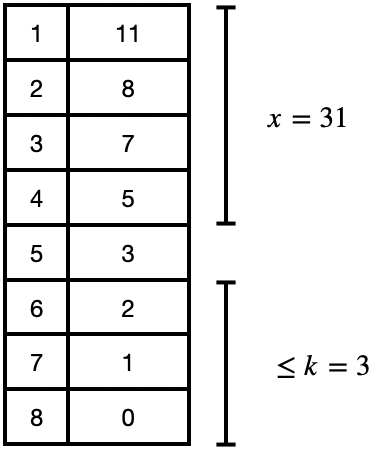}
        \caption{Maximum number of non-zeros below a 3-boundary bucket.}
    \end{subfigure}\hspace{0.5cm}
    \begin{subfigure}{0.4\linewidth}
        \centering
        \includegraphics[scale=0.5]{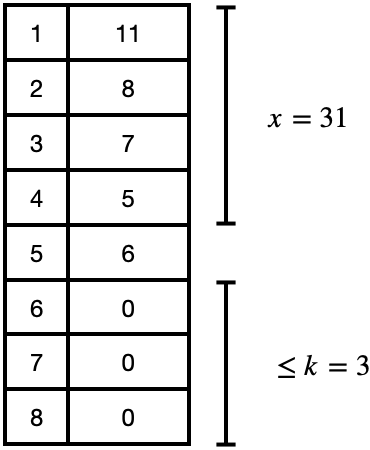}
        \caption{Minimum number of non-zeros below a 3-boundary bucket.}
    \end{subfigure}
    
    \caption{Example of a $k$-boundary bucket (where $k=3$) for a sketch column with bucket depth and the size of each bucket. In this example, bucket 6 is the boundary bucket and the total number of indices in the column is 36. The total number of vector indices with depth 1 through 4 is $x=31$. Given that $b_6$ is the boundary bucket and $x=31$ (the number of indices in $b_4$ or above), the number of vector indices with depth 5 is in range $[3, 6]$.}
    \label{fig:boundary}
\end{figure}

We construct our sketches using a \defn{$k$-wise independent} hash function. The size of a bucket $|b_i|$ is equal to the number of nonzero indices that hash to depth $i$. For \sketchname this is also equivalent to the number of indices held by the bucket. $|b_i|$ is a random variable that depends upon the hash function and the set of nonzeros. Let \defn{$Z$} be the set of nonzero indices in the input vector and let \defn{$z$} $ = |Z|$. The maximum depth (number of rows) of the column is $d = O(\log n)$.

The \defn{k-boundary bucket} (for a constant $k$) is the bucket of smallest depth for which $\leq k$ non-zero vector indices have depth greater than or equal to the boundary. That is the boundary bucket is $\min_i \sum_{j=i}^d |b_j| \leq k$.

Our analysis will assume that $z > k$ as otherwise we can directly apply the fully independent results that we will get to below.

\begin{lemma}
    \label{lem:boundary}
    If bucket $\buck_{\boundaryidx+1}$ is a $k$-boundary bucket, the sizes of the buckets $[\boundaryidx, d]$ are determined by the outcome of $k$ random variables. 
    Specifically, let $t = \sum_{j=1}^d |b_j|$ be the total number of vector indices that appear in the entire column (but not the deterministic bucket). If bucket $\buck_{\boundaryidx+1}$ is a $k$-boundary bucket and the number of indices that map to buckets above $b$ is $x = \sum_{j=1}^{\boundaryidx-1} |b_j|$, then there exists a set $F \subseteq Z$ such that $|F| = k$, $|b_\boundaryidx \cap F| + t - x = S_\boundaryidx(Z)$, and $\forall i \in [\boundaryidx+1, d]$, $|b_i \cap F| = |b_i|$. 
\end{lemma}
\begin{proof}
    By construction, bucket $\buck_\boundaryidx$ has size of at least $t-x-k$ and at most $t-x$. If this were not the case, then bucket $\buck_{\boundaryidx+1}$ could not be the boundary bucket. Assume for the sake of contradiction that this condition does not hold. Thus, $|b_\boundaryidx|$ must either by less than $t-x-k$ or greater than $t-x$. The case of $|b_\boundaryidx| > t-x$ is straightforwardly impossible as $x + |b_\boundaryidx|$ would be greater than $t$. In the case where $|b_\boundaryidx| < t - x - k$, then either $|b_\boundaryidx|$ is negative (an impossibility) if $x + k > t$ or $\sum_{j=\boundaryidx+1}^d |b_j| > k$ thus violating the assumption that $\buck_{\boundaryidx+1}$ is a $k$-boundary bucket. 

    Given that we know that $|b_\boundaryidx|$ is the sum of $z$ indicator random variables and the the value of $|b_\boundaryidx|$ can only vary by at most $k$, there must exist a set of $k$ indicator random variables whose outcome defines $|b_\boundaryidx|$. Additionally, since at most $k \leq t - x - |b_\boundaryidx| \geq 0$ non-zeros may appear in buckets $[\boundaryidx+1, d]$ by construction, there must exist a set of $k$ non-zeros whose depth determines the size of all buckets $[\boundaryidx, d]$.
\end{proof}

\begin{figure}[t]
    \centering
    \begin{subfigure}[!t]{0.4\linewidth}
        \centering
        \includegraphics[scale=0.5]{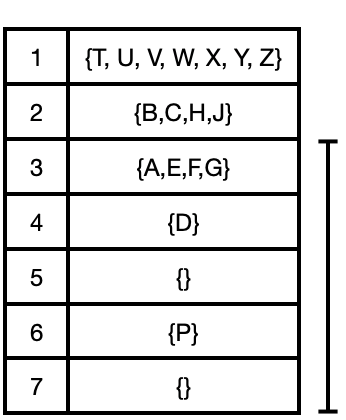}
        \caption{Original column with 3-boundary bucket $B_4$. The indicated buckets will be isolated.}
    \end{subfigure}\hspace{0.5cm}
    \begin{subfigure}[t]{0.4\linewidth}
        \centering
        \includegraphics[scale=0.5]{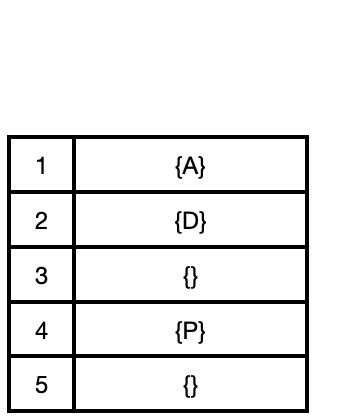}
        \caption{A $3$-isolated column. Notice that the first bucket contains only a single non-zero. However, in the original column this is a bad bucket. This is an example of why we always assume that $B_1$ in the isolated column is bad.}
    \end{subfigure}
    
    \caption{An example of creating an isolated column. Notice that the isolated column contains only three non-zeros bucket we preserve a mapping between good buckets in the isolated column and those in the original column.}
    \label{fig:isolated}
\end{figure}

We construct a \defn{$k$-isolated column}: a sketch column that contains only the set $F$ (that we defined in Lemma~\ref{lem:boundary}) of size $k$ non-zeros and the buckets in which they could appear. That is, we isolate a contiguous subset of the column's buckets and a subset of the non-zero indices. 
If the $k$-boundary bucket of a column $C$ is $\buck_{\boundaryidx+1}$, then the $k$-isolated column $\hat{C}$
contains buckets $\buck_b, \dots, \buck_d$.

Additionally, a $k$-isolated column is good if and only if a bucket of depth greater than one is good. That is, we assume bucket $\buck_1$ is always bad. An example of a column and the associated isolated column can be found in Figure~\ref{fig:isolated}.

\begin{lemma}
\label{lem:iso_correctness}
Given bucket $\boundaryidx + 1$ is the $k$-boundary bucket for the original column $C$,
then the $k$-isolated column $\hat{C}$'s bucket $\hat{\buck}_i$ 
for $i \in [2, d - \boundaryidx + 1]$ is good if and only if bucket $B_{i+\boundaryidx-1}$ is good in $C$.
\end{lemma}
\begin{proof}
    With the exception of the first bucket $\buck_1$ and the deterministic bucket $\buck_0$, the isolated column has buckets of identical content to those of the original column $C$. Thus, bucket $\hat{\buck}_i$ in the isolated column has identical content to bucket $\buck_{\boundaryidx+i-1}$ for $i\in[2, d - \boundaryidx + 1]$, therefore, Bucket $\hat{\buck}_i$ is good if and only if bucket $\buck_{\boundaryidx+i-1}$ is good.
\end{proof}

We now prove that we can analyze the $k$-isolated column without needing to condition on the state of the original column. That is, we can view it as an entirely new column that only contains $k$ nonzero elements.

\begin{lemma}
\label{lem:iso_prob_dist}
Let column $C_h$ have depth $d$ and be constructed with a $k$-wise independent hash function $h$ with image $[2^d]$, and let $C_h$'s $k$-isolated column ``begin'' at an index $\boundaryidx$.
Then the probability distribution of the size of each bucket in the $k$-isolated column function is identical to that of a column $C_{h'}$ with depth $d-\boundaryidx+1$.
\end{lemma}
\begin{proof}
    Let bucket $\buck_{\boundaryidx+1}$ be the $k$-boundary bucket for the original column $C_h$ (with $h: [n] \to [2^d]$).
    The $k$-isolated column's 
    hash function can thus defined as 
    $$\hat{h}(i) = \lfloor h(i) / 2^{\boundaryidx-1} \rfloor$$
    Note that, by the definition of $k$-wise independent hashing, for all $y \in [2^d]$,
    and for all $x \in [n]$
    $$\Pr[h(x) = y ] = \frac{1}{2^d}.$$
    It follows directly from this that 
    \begin{align*}
        \Pr [\hat{h}(x) = y ] &= \Pr[\lfloor h(x) / 2^{\boundaryidx-1} \rfloor = y] \\
        &=  \sum_{i = 0}^{2^{b-1} - 1} \Pr[h(x) = y \cdot 2^{\boundaryidx-1} + i] \\
        &=  \sum_{i = 0}^{2^{b-1} - 1} \frac{1}{2^d} = \frac{1}{2^{d-\boundaryidx+1}}\text{,}
    \end{align*}
    for all $x$ and $y$. 
    Therefore, the distribution of a single hash call is uniform. Likewise, note that $\hat{h}$ is still $k$-wise
    independent since functions of $k$-wise independent variables are also $k$-wise 
    independent. Thus, $\hat{h}$ satisfies the desired properties of a
    $k$-wise independent hash function, defined over $[n] \to [2^{d-\boundaryidx+1}]$. 

    Therefore, the sizes of the buckets constructed using $\hat{h}$ have an equivalent distribution to buckets constructed from hash function $h' : [n] \to [2^{d-\boundaryidx+1}]$ that is drawn independently (from the choice of $h$)
    from a family of $k$-wise independent hash functions.
    The probability distributions of these hash functions on the input to the $k$-isolated column are identical, that is, $\forall i \in [1, d-\boundaryidx+1]$, $P[\log(\hat h(x)) = i] = P[\log(h'(x)) = i] = P[\log(h(x)) = i + \boundaryidx - 1 \mid \log(h(x)) \geq \boundaryidx]$. This result follows from the above
    statements about the distribution of $\hat{h}$.
    This statement holds for all inputs to the $k$-isolated column as by definition their depth must be at least $\boundaryidx$ in the original column.
\end{proof}

\paragraph{Fully Independent Hash Functions}
The success probability of \sketchname constructed with a fully independent hash function with $z > 1$ nonzeros (if $1$ nonzero then the column succeeds deterministically) and column depth $d > 0$ is given by the function $F(z, d)$. Where $F$ is defined by the following recurrence:
\[F(a, b) =
\displaystyle\sum_{i \in [0, a] \setminus \{1\}} \left(2^{-a}\binom{a}{i} F(a - i, b - 1)\right) + a 2^{-a}
\]
and if $a \leq 0$ or $b \leq 0$ then the probability of success is $0$.
For this recurrence, we look at the bucket at depth $1$, and consider every possible number of items that can
go in that first bucket and the probability of that scenario occurring. For each of these we calculate the likelihood that a bucket
below is correct recursively in the case that the first bucket does not have exactly one element.

If we exclude the first bucket being good from our definition of
success (which is required for our analysis of a $k$-isolated column), then the probability that the column succeeds is $\hat{F}(z, d) = F(z, d) - z 2^{-z} \cdot (1 - F(z - 1, d - 1))$. Specifically, we are accounting for the possibility that the first bucket is good: $z 2^{-z}$, and the rest of the buckets are bad: $1 - F(z-1, d-1)$.

For small values of $z$, we can solve this recurrence manually. Furthermore,
note that though our hash functions are taken to have limited $k$-wise 
independence, as long as $z$ is less than our choice of $k$, then analysis 
with these recurrences still holds.

Additionally, $F$ is non-decreasing with respect to $d$, that
is for any $z$, $F(z, d_1) \geq F(z, d_2)$ as long as $d_1 \geq d_2$.
Since $\forall z, d$,  $F(z, d) \geq 0$, it must be the case that adding additional subproblems (by increasing $d$) without changing the existing subproblems can only maintain or increase the likelihood of success.

\paragraph{Extending to k-wise independence}
Now we prove our theorems that analyze sketches constructed with $k$-wise independent hash functions. Our key strategy is to use $k$-isolated columns and the associated lemmas to reduce to a column on only $k$ non-zero elements. Thus, we can perform fully independent analysis on this subproblem to solve for the probability the entire column succeeds.

\begin{table}[t]
    \centering
    \begin{tabular}{r|r}
       Non-zero Indices  &  \sketchname \\
       \hline
       $1$  &  $1$\\
       $2$  &  $0.666$ \\
       $3$  &  $0.856$\\
       $4$  &  $0.799$\\
       $5$  &  $0.813$\\
       $6$  &  $0.810$\\
       $7$  &  $0.810$
    \end{tabular}
    \caption{Lower bound on the probability that a \sketchname column (with 10 buckets) succeeds given a small number of non-zero indices, and with full independence.}
    \label{tab:cameo_exact_pr}
\end{table}

\begin{lemma}
\label{lem:3-wise-sketch}
With $k \geq 3$-wise independent hash functions \sketchname sketch columns have a success probability of at least $0.66$.
\end{lemma}
\begin{proof}
    If the number of non-zeros is at most 3, we can solve for the probability of success assuming full independence. Table~\ref{tab:cameo_exact_pr} gives the likelihood of success for 1, 2 and 3 non-zeros which are all greater than $0.66$.

    Otherwise, we analyze a $3$-isolated column to lower bound the success probability of a sketch column $C_h$ constructed from a $3$-wise independent hash function $h$ when $z > 3$. 
    We will assume $n \geq 2^{10}$ (and artificially increase the size of the vector and thus sketch if it is not). Additionally, let the depth of the column be $d = \log n + 5$.
    
    By Lemma~\ref{lem:boundary} and Lemma~\ref{lem:iso_correctness}, we can analyze $3$ non-zeros to lower bound the success probability of a sketch on $z > 3$ non-zeros. Then, by Lemma~\ref{lem:iso_prob_dist} we can perform this analysis assuming these $3$ non-zeros exist within an isolated column of depth $d-b+1$. This $3$-isolated column has a depth function with probability distribution equivalent to that of a regular column with depth $d-\boundaryidx+1$.

    Solving $\hat{F}$ for $z = 3$ (i.e., with the additional assumption that the first bucket is always bad) 
    gives a success probability of greater than $0.69$ for \sketchname if $d \geq 5$.
    
    We also need to account for the probability that the $k$-boundary bucket is too far down
    in the sketch column,
    that is $d-\boundaryidx+1$ is too small to give a high likelihood of success. 
    To do so, we assume that if $\boundaryidx$ is of too high depth, then the probability of success is zero. If $\boundaryidx$ is of depth less than or equal to $\log n + 1$, then the analysis holds.
    
    The probability that the boundary is greater than this cutoff is $1-\binom{z}{3} 2^{-3(\log n + 2)}$. Recall that $n \geq 2^{10}$, at this vector size the likelihood that the boundary occurs at the depth $\log n + 2$ bucket or less is $0.98$ and so the overall probability of success is $0.69 \cdot 0.98 > 2/3$.

    Therefore, \sketchname column on a vector of length $n$ with $\log n + 5$ buckets (to assure that $d-\boundaryidx+1$ is at least $5$) has a success probability of at least $2/3$.
\end{proof}

This improved probability of success per column means that our sketches can be much more compact while retaining the probability guarantees we require. Thus, we can finally prove Theorem~\ref{thm:cameo_cols}.



\restatecameospacethm
\begin{proof}
    With $3$-wise independence each \sketchname column succeeds with probability at least $0.66$ by Lemma~\ref{lem:3-wise-sketch}. Thus, the probability that a single column fails is $0.33$. So to achieve a failure probability of $\delta$ requires $\log_{1/0.33} (1 / \delta) \leq \log_{3} (1 / \delta)$ columns.
\end{proof}

We do not know if \sketchname samples nonzero elements uniformly. For the purposes of this paper it makes no difference as these connectivity algorithms require only some element from the set. They do not require a uniformly random element. 

We can make the following simple modification to \sketchname to ensure it samples uniformly. We only count a bucket as good if it contains one nonzero and is the deepest nonempty bucket. Thus, we simulate \cubesketch's query procedure but retain the improved update time. Our improved analysis of the space constants also holds but requires more detail that we exclude here.

\clearpage
\end{document}